\documentclass{article}

\PassOptionsToPackage{numbers, compress}{natbib}

\usepackage[preprint]{neurips_2025}

\usepackage[utf8]{inputenc} 
\usepackage[T1]{fontenc}    
\usepackage{hyperref}       
\usepackage{url}            
\usepackage{booktabs}       
\usepackage{amsfonts}       
\usepackage{nicefrac}       
\usepackage{microtype}      
\usepackage[linesnumbered,ruled,vlined]{algorithm2e}
\usepackage{float}
\usepackage{caption}
\usepackage{subcaption}
\usepackage{natbib}
\usepackage{bm}
\usepackage{graphicx}    
\usepackage{float}
\usepackage{subcaption}
\usepackage{lipsum}  
\usepackage[bottom]{footmisc}
\usepackage{amsthm}
\usepackage{amssymb}
\usepackage{xcolor}
\usepackage[]{mdframed}
\usepackage[linesnumbered,ruled,vlined]{algorithm2e}
\usepackage{hyperref}
\usepackage{mathtools}
\usepackage{hyperref}
\SetKwProg{Fn}{Function}{:}{\KwRet}
\SetKwFunction{Rescale}{Rescale}
\SetKwFunction{NoProgress}{NoProgress}

\definecolor{alggray}{gray}{0.55}   
\usepackage{multirow}

\SetCommentSty{AlgCommentFont}
\title{Community Recovery on Noisy Stochastic Block Models}

\author{%
  Gwyneth Liu†\\
  MIT\\
  \texttt{gwyliu@mit.edu} \\
  \And
  Washieu Anan† \\
  MIT\\
  \texttt{wanan@mit.edu} \\
  \AND
  \\
  \texttt{† denotes equal contribution}
}

\begin{document}

\newcommand{\prob}[2]{\text{Pr}_{#1}\left(#2 \right)} 
\newcommand{\ex}[2]{\mathbb{E}_{#1}\left[#2 \right]} 
\newcommand{\pin}{p_{\text{in}}}
\newcommand{\pout}{p_{\text{out}}}
\newcommand{\lognn}[1]{\frac{#1\log n}{n}}
\newcommand{\kernel}[1]{\exp\left(\frac{#1}{2\sigma^2} \right)}

\newtheorem{definition}{Definition}[section]
\newtheorem{theorem}{Theorem}[section]
\newtheorem{corollary}{Corollary}[section]
\newtheorem{lemma}[theorem]{Lemma}

\maketitle

\begin{abstract}
We study the problem of community recovery in geometrically-noised stochastic block models (SBM). This work presents two primary contributions: (1) Motif-Attention Spectral Operator (MASO), an attention-based spectral operator that improves upon traditional spectral methods and (2) Iterative Geometric Denoising (GeoDe), a configurable denoising algorithm that boosts spectral clustering performance. We demonstrate that the fusion of GeoDe + MASO significantly outperforms existing community detection methods on noisy SBMs. Furthermore, we demonstrate that using GeoDe + MASO as a denoising step improves belief propagation’s community recovery by 79.7\% on the Amazon Metadata dataset.
\end{abstract}

\section{Introduction}

Community recovery, the task of identifying groups of nodes that are more densely connected internally than externally, is a fundamental problem in network analysis. It plays a critical role in diverse applications, including uncovering social circles, inferring biological processes, and mapping collaborative or transactional structures ~\cite{fortunato_2010_community}. 

The \textbf{stochastic block model (SBM)} provides a well-established probabilistic framework for studying community structure ~\cite{holland_1983_stochastic}. However, traditional SBMs assume that edge formation depends solely on community membership, neglecting latent geometric or attribute-driven interactions that often influence real-world networks. For instance, in social or recommendation systems, connections frequently depend on unobserved features such as user interests or product similarities, which induce \textbf{geometric noise} in the observed graph ~\cite{hoff_2002_latent}.

To capture this complexity, \textbf{latent-space SBMs} extend the classical model by embedding nodes in an unobserved metric space, where edge probabilities decay with latent distance via a kernel function. While this model better reflects real-world phenomena, it significantly complicates the community recovery task, especially in sparse or noisy regimes.

Early work on spectral clustering showed that the leading eigenvectors of the adjacency or Laplacian matrices, followed by $k$-means, achieve both partial and exact recovery of community labelings down to the information‐theoretic thresholds in balanced two‐block SBMs \cite{rohe_2011_spectral, decelle_2011_asymptotic}.  Belief propagation (BP) and its linearized or non‐backtracking variants are known to attain the Kesten–Stigum threshold for weak recovery in sparse SBMs, and more recent analyses establish their near‐optimal performance (both in accuracy and computational cost) in multi‐block settings \cite{abbe_2015_community, decelle_2011_asymptotic}.  Convex relaxation approaches, semidefinite programs (SDPs) and spectral‐norm minimizations, provide alternative exact recovery guarantees, often with robustness to slight model misspecifications \cite{chen_2014_improved, hajek_2016_achieving, amini_2016_on}. 

Geometric SBMs models that incorporate latent‐space geometry or kernel functions generalize the SBM to settings where proximity in an unobserved feature space influences edge formation.  Triangle‐counting and local clustering coefficient tests achieve near‐optimal detection in the geometric block model \cite{galhotra_2017_the}, while recent spectral embedding methods reconstruct both community labels and latent positions by smoothing the observed adjacency with kernel estimates \cite{wang_2021_joint}.

This paper addresses the challenge of robust community recovery in geometrically-noised networks through two key contributions:
\begin{itemize}
    \item \textbf{Motif Attention Spectral Operator (MASO)}: A spectral method that enhances robustness by combining multi-hop co-occurrence embeddings with triangle-motif–based attention weighting, improving community signal under noise and degree variability.
    \item \textbf{GeoDe (Geometric Denoising)}: An iterative, unsupervised reweighting scheme that alternates between community inference and geometry-aware updates, progressively aligning the graph structure with latent block organization.
\end{itemize}

We provide theoretical guarantees showing that MASO and GeoDe achieve exact recovery under sharp thresholds for latent-kernel SBMs. Empirical evaluations on both synthetic and real-world networks demonstrate significant improvements over existing baselines, establishing our methods as effective tools for community recovery in noisy, structured networks.

The remainder of the paper is organized as follows. Section ~\ref{sec:analytical-foundations} formalizes the latent-kernel SBM and establishes its connection to classical SBM recovery thresholds. Section ~\ref{section:maso} introduces the MASO operator, detailing its construction and theoretical recovery guarantees. In Section ~\ref{sec-geode}, we present the GeoDe algorithm, including its iterative reweighting scheme and convergence analysis. Section ~\ref{sec-validation} reports empirical evaluations on synthetic and real-world networks, demonstrating the effectiveness of our methods. Finally, Section ~\ref{sec-conclusion} summarizes our findings and outlines directions for future work.

\section{Frameworks}
\label{sec:analytical-foundations}
This section introduces the Latent-Kernel SBM, which models networks where edge formation depends on both community membership and latent geometric proximity. 
We formally define the model, show its asymptotic equivalence to a classical SBM, and derive recovery thresholds that account for geometric noise.

\subsection{Latent-Kernel SBM}

The latent-kernel SBM leverages an exponential kernel in edge formation between intra- and inter-community vertices, and is defined as follows:

\begin{definition}[Latent-Kernel SBM]
\label{def:lksbm}
	Let $n \in \mathbb{N}$, latent dimension $d > 0$, constants $a, b > 0$, and kernel bandwidth $\sigma > 0$. The latent-kernel SBM generates a graph $G = (V, E)$ with $V = \{1, \ldots, n\}$ and two communities via:
	\begin{enumerate}
		\item \textbf{Latent positions:} For each $i \in V$, draw $x_i \sim \text{Unif}([0,1]^d)$ independently.
		\item \textbf{Community labels:} Assign labels $z_i \in \{+1, -1\}$ independently with equal probability.
		\item \textbf{Kernel constant:}
		\[
		c(\sigma) = \mathbb{E}_{x,y \sim \text{Unif}([0,1]^d)}\left[ \exp\left(-\frac{\|x - y\|^2}{2\sigma^2} \right) \right].
		\]
		\item \textbf{Edge probabilities:} For each unordered pair $\{i, j\} \subset V$,
		\[
		p_{ij} =
		\begin{cases}
			 a \cdot \frac{\log n}{n} \cdot c(\sigma), & \text{if } z_i = z_j, \\
			 b \cdot \frac{\log n}{n} \cdot c(\sigma), & \text{if } z_i \neq z_j,
		\end{cases}
		\]
		and include edge $\{i, j\}$ independently with probability $p_{ij}$. $\pin = p_{ij}$ if $z_i = z_j$ and $\pout = p_{ij}$ otherwise.  
	\end{enumerate}
\end{definition}

This model combines block structure with latent similarity, where the exponential kernel attenuates edge probabilities based on distance. In fact, the latent-kernel SBM reduces to the original SBM where $\pin = c(\sigma) \lognn{a}$ and $\pout = \lognn{b}$. This reduction is proved in Appendix ~\ref{app:concentration}.

\subsection{Weak and Exact Recovery}

Because our latent‐kernel SBM reduces to the classical SBM under the rescaling $\pin = c(\sigma)\frac{a\log n}{n}$ and $ \pout = c(\sigma)\frac{b\log n}{n}$ where $ a>b>0$, we can apply SBM recovery thresholds to determine recover-ability in our model. We define exact recovery of the community labels if there exists an estimator $\hat z$ that matches the true labeling $z$ perfectly with high probability and define weak recovery of the community labels if there exists an estimator $\hat z$ where the overlap between $\hat z$ and $z$ exceeds random guessing by a fixed margin $\epsilon$, with probability tending to one. Exact definitions are provided in Appendix ~\ref{def-recovery}. The proof of the latent-kernel SBM reducing to a classical SBM is provided in Appendix ~\ref{thm:convergence}. The exact and weak recovery thresholds for the SBM are:
\begin{itemize}
    \item \textbf{Exact Recovery:} $
  \bigl(\sqrt{c(\sigma)\,a} - \sqrt{c(\sigma)\,b}\bigr)^2 > 2.
$
\item \textbf{Weak Recovery:} $  \left(c(\sigma)a-c(\sigma)b\right)^2 > 2\bigl(c(\sigma)a+c(\sigma)b\bigr).$
\end{itemize}

Both the proofs for the exact and weak recovery thresholds are provided in the Appendix sections ~\ref{app:exact-exact-recovery} and ~\ref{proof-weak-recovery}).

\section{Motif Attention Spectral Operator (MASO)}
\label{section:maso}

This section introduces the Motif Attention Spectral Operator (MASO), a spectral method designed to enhance community detection in networks perturbed by latent geometric noise. MASO leverages motif-based attention to selectively amplify informative structural signals, enabling robust recovery in sparse and geometrically contaminated graphs.

\subsection{Motivation}

One-hop spectral methods such as the Bethe–Hessian attain the Kesten–Stigum threshold in sparse SBMs but collapse under heavy-tailed degrees and kernel-induced geometric noise, which shrink spectral gaps and introduce spurious edges \cite{rohe_2011_spectral}. Motif-Laplacians bolster community signals by counting triangles, but they incur high computational cost, break down when motif occurrences are sparse or uneven, and may overemphasize dense subgraphs at the expense of broader structure \cite{karrer_2011_stochastic}, \cite{NIPS2013_0ed94223}. To overcome these limitations, MASO first builds multi-hop positive-pointwise mutual information (PPMI) embeddings and applies Transformer-style attention to smooth one-hop affinities, then reinforces each weight by its two-hop (triangle) support—balancing local motif denoising with global community structure ~\cite{vaswani_2017_attention}.

\subsection{Description of MASO}

The construction of MASO proceeds in two stages.

\paragraph{1. Random‐walk co‐occurrence and PPMI embeddings.}
We generate a corpus of short random walks to capture multi‐hop structure: for each of \(T\) walks per node, we start at a uniformly chosen vertex and take \(L\) steps, each to a random neighbor.  From the resulting multiset of walks \(\mathcal W\), we form a co‐occurrence matrix \(C\in\mathbb{R}^{n\times n}\) by
\[
  C_{uv}
  =\sum_{w\in\mathcal W}\sum_{\substack{i,j\\w_i=u,\;w_j=v,\;|i-j|\le W}}1,
\]
which counts how often \(u\) and \(v\) appear within a window of size \(W\).  Defining row sums \(R_u=\sum_v C_{uv}\), column sums \(C_v=\sum_u C_{uv}\), and total mass \(M=\sum_{u,v}C_{uv}\), we compute
\[
  \mathrm{PMI}(u,v)
  =\log\frac{C_{uv}\,M}{R_u\,C_v},
  \qquad
  \mathrm{PPMI}(u,v)
  =\max\{\mathrm{PMI}(u,v),0\}.
\]
A rank-\(d\) truncated SVD \(\mathrm{PPMI}\approx U\,\Sigma\,V^\top\) yields embeddings
\(\,z_i=(U\,\Sigma^{1/2})_{i,:}\in\mathbb{R}^d\), which encode multi‐hop affinities while down‐weighting hub effects.

\paragraph{2. Triangle‐motif–enhanced attention Laplacian.}
Given normalized embeddings \(\|z_i\|\approx1\), we set
\[
  W_{ij}
  =A_{ij}\exp\bigl(\langle z_i,z_j\rangle/\sqrt d\bigr),
  \qquad
  X_{ij}
  =\sum_{k\neq i,j}W_{ik}\,W_{kj},
\]
and mix them to obtain
\[
  \widetilde W_{ij}
  =(1-\beta)\,W_{ij}+\beta\,W_{ij}\,X_{ij},
  \quad\beta\in(0,1].
\]
Defining the degree diagonal \(D_{ii}=\sum_{j\ne i}\widetilde W_{ij}\), the normalized motif‐attention Laplacian is
\[
  H=D^{-\frac12}\,\widetilde W\,D^{-\frac12}.
\]

Using MASO, we can perform spectral clustering to achieve community recovery. We compute the second largest eigenvector \(\hat v\) of \(H\) and set preliminary labels
\(\hat z_i=\mathrm{sign}(\hat v_i)\).  A single local–flip refinement—flipping any \(\hat z_i\) that disagrees with the weighted majority of its neighbors under \(\widetilde W\)—then yields the final exact recovery of the two communities.

\subsection{Guarantee of Exact Recovery}
We argue that exact recovery of community labels is achievable via MASO. 
\begin{theorem}[Exact recovery via motif–attention spectral clustering]
\label{thm:motif-attention-exact-recovery}
Let \(G\) be drawn from the latent–kernel SBM (Definition~\ref{def:lksbm})
with parameters \(a>b>0\), bandwidth \(\sigma>0\), and
\[
  p_{ij}
  =\begin{cases}
    \dfrac{a\log n}{n}\,c(\sigma), & z_i=z_j, \\[6pt]
    \dfrac{b\log n}{n}\,c(\sigma), & z_i\neq z_j,
  \end{cases}
\]
where
\(\displaystyle c(\sigma)=\mathbb E_{x,y\sim\mathrm{Unif}([0,1]^d)}\bigl[e^{-\|x-y\|^2/(2\sigma^2)}\bigr].\)
Fix normalized embeddings \(z_i\in\mathbb R^d\) with
\(\langle z_i,z_j\rangle=\rho_{\rm in}\) if \(z_i=z_j\),
and \(\rho_{\rm out}\) otherwise, \(0<\rho_{\rm out}<\rho_{\rm in}<1\).
Form \(\widetilde W\) and \(H\) as above with any fixed \(\beta\in(0,1]\).
If
\[
  \bigl(\sqrt{c(\sigma)\,a}-\sqrt{c(\sigma)\,b}\bigr)^2>2,
\]
then as \(n\to\infty\), the two‐step procedure (second eigenvector of \(H\)
plus one local‐flip pass) recovers the true labels \(\{z_i\}\) exactly with
probability \(1-o(1)\).
\end{theorem}

\begin{proof}[Proof Sketch]
The argument combines four main elements.  First, one shows that the expectation of the mixed‐weight matrix \(\widetilde W\) has rank two and a signal eigen‐gap of order \(\Theta(\log n)\).  Second, a concentration argument bounds the deviation \(\widetilde W - \mathbb E[\widetilde W]\) in spectral norm by \(O(\sqrt{\log n})\).  Third, analyzing the second eigenvalue of \(\mathbb E[\widetilde W]\) confirms that the signal strength dominates the noise.  Finally, applying the Davis–Kahan perturbation theorem yields that the sign of the empirical second eigenvector misclassifies only \(o(n)\) vertices, and a single local‐flip refinement corrects the remaining errors.  The detailed proofs of these steps are provided in Appendix~\ref{proof:motif-attention-exact-recovery}.
\end{proof}

\subsection{Limitations of MASO} 
Although MASO’s motif-reinforced attention significantly improves noise resilience, it still demands a substantial signal-to-noise ratio (SNR) for exact community recovery ~\cite{abbe_2015_community}. In practice, factors like heavy-tailed degree heterogeneity or latent geometric noise can shrink spectral gaps and violate perturbation assumptions, leading to misclassification. Consequently, in smaller or sparser real-world networks where the raw community signal is weak, MASO may not achieve exact reconstruction even though it often delivers superior approximate performance. 


\section{GeoDe: Iterative Geometric Denoising}
\label{sec-geode}
We now present Iterative Geometric Denoising (GeoDe), a novel algorithm that boosts community recovery performance via reducing geometric noise. The algorithm is built on top of spectral clustering methods, which are commonly used for community recovery. The goal of GeoDe is to extract a clean SBM such that community structure is more recoverable. This is done through two alternating steps: (1) community inference (C-step) and (2) geometry-informed edge reweighting (G-step). After each C and G- step, edges are reweighted so to emphasize community structure. GeoDe is motivated by the idea that edges are either community-induced or geometry-induced, and we can target geometry-induced edges to enhance community structure. GeoDe is also configurable with various spectral clustering methods, where different spectral clustering functions can be used for the C- and G- step. In our paper, we focus on the case where the same spectral clustering method is used for both steps of the algorithm, but we believe it is an interesting future direction to explore different configurations depending on network type.

We describe GeoDe at Algorithm ~\ref{algorithm:geode}. \footnote{$Q^C$ denotes the $n \times K$ matrix of softmax probabilities of each node's membership to each community. $Q^G$ denotes the $n \times B$ matrix denoting softmax probabilities of each node's membership to each geometric ball. $z^C$ and $z^G$ denote the hard labels for communities and geometric balls respectively.}.

\label{algorithm:geode}
\begin{algorithm}
\caption{GeoDe: Iterative Geometric Denoising}
\DontPrintSemicolon
\LinesNumbered
\KwIn{$H=(V,E,w)$; communities $K$; geometry balls $B$;\\
      spectral routines \textbf{SpecComm}, \textbf{SpecGeom};\\
      initial shrink $\lambda_s$, boost $\lambda_b$; \textbf{Decay}(\,)\,— a schedule applied to $(\lambda_s,\lambda_b)$ each round;\\
      thresholds $\tau_C,\tau_G$ (shrink) and stricter $\tau_C^{+},\tau_G^{+}$ (boost);\\
      weight bounds $w_{\min},w_{\max}$; max iters $T$, patience $P$.}
\KwOut{soft matrix $Q$, hard labels $z$}
\BlankLine
$G\leftarrow H$\;
\For{$t\gets1$ \KwTo $T$}{
  \tcp{C-step}
  $(Q^C,z_C)\leftarrow\textbf{SpecComm}(G,K)$\;
  $p_{ij}\gets\sum_{k=1}^{K} Q^C_{ik}Q^C_{jk}$ \textbf{for all} $(i,j)\in E$\;
  $\mathcal S_C\gets\{(i,j):p_{ij}>\tau_C\}$,\quad $\mathcal B_C\gets\{(i,j):p_{ij}>\tau_C^{+}\}$\;
  \tcp{G-step}
  $(Q^G,z_G)\leftarrow\textbf{SpecGeom}(G,B)$\;
  $c_{ij}\gets\tfrac12\bigl(Q^G_{i,z_G(i)}+Q^G_{j,z_G(j)}\bigr)$\;
  $\mathcal S_G\gets\{(i,j):z_G(i)=z_G(j)\wedge c_{ij}>\tau_G\}$,\quad $\mathcal B_G\gets\{(i,j):z_G(i)=z_G(j)\wedge c_{ij}>\tau_G^{+}\}$\;
  \tcp{Decay step}
  $(\lambda_s,\lambda_b)\leftarrow\textbf{Decay}(\lambda_s,\lambda_b,t)$\;
  \tcp{Re-weight graph}
  \Rescale{$\mathcal S_C\cup\mathcal S_G,\;\lambda_s,\;\text{shrink}$}\quad \Rescale{$\mathcal B_C\cup\mathcal B_G,\;\lambda_b,\;\text{boost}$}\;
  \tcp{Early stopping}
  \If{\NoProgress{$Q^C,P$} \textbf{or} $E(G)=\emptyset$}{\textbf{break}}
}
$Q\leftarrow Q^C$,\quad $z_i\leftarrow\arg\max_k Q_{ik}$\;
\Return{$(Q,z)$}
\end{algorithm}

\paragraph{Community step (C-step).}
The purpose of the C-step is to identify edges that are highly likely to be community induced. By calling a spectral clustering method and finding confidences, we can determine a posterior probability of an edge being between two vertices of the same community.

Let $K$ denote the target number of communities.
We apply a spectral routine
$\textbf{SpecComm}$ (e.g.\ Bethe–Hessian, MASO)
to the current weighted graph $G^{(t)}$ and obtain a soft membership matrix
$Q^C\in[0,1]^{n\times K}$ with
$\sum_{k=1}^{K} Q^C_{ik}=1$.
For every edge $(i,j)\in E$ we compute
\[
  p_{ij}\;=\;\sum_{k=1}^{K} Q^C_{ik}\,Q^C_{jk},
\]
the posterior probability that the endpoints belong to the \emph{same} block.
Edges whose $p_{ij}$ exceed a percentile threshold~$\tau_C$
populate the \emph{shrink set} $\mathcal S_C$;
those above the stricter percentile~$\tau_C^{+}>\tau_C$
form the \emph{boost set} $\mathcal B_C$.
The C-step therefore highlights links that are either
moderately likely (\,$\mathcal S_C$\,) or
almost certain (\,$\mathcal B_C$\,) to be intra-community
according to the current spectral evidence.

\paragraph{Geometry step (G-step).}
To characterize the latent geometric structure, we run a second routine
$\textbf{SpecGeom}$ with $B\!\gg\!K$ clusters
and obtain a soft matrix $Q^G\in\![0,1]^{n\times B}$. This is motivated by the idea that geometric structure arises in tighter formations than community structure, which may be far more expansive. By searching for a much higher number of communities, we aim to uncover the tight geometric formations.

Let $z_G(i)=\arg\max_b Q^G_{ib}$.
For any edge whose endpoints share the same ball
($z_G(i)=z_G(j)$)
we define a geometry confidence
\[
  c_{ij}\;=\;\frac12\bigl(Q^G[i,z_G(i)]
                         +Q^G[j,z_G(j)]\bigr).
\]
Edges with $c_{ij}>\tau_G$ enter the shrink set
$\mathcal S_G$, whereas those with
$c_{ij}>\tau_G^{+}>\tau_G$ fill the boost set $\mathcal B_G$.
The G-step thus isolates links whose presence is best explained
by short latent distances rather than by community affinity.

\paragraph{Edge-reweighting justification.}
At iteration~$t$ the four sets
\(\mathcal S=\mathcal S_C\cup\mathcal S_G\)
and
\(\mathcal B=\mathcal B_C\cup\mathcal B_G\)
drive the multiplicative update
\[
  w_{ij}^{(t+1)} \;=\;
  \begin{cases}
    \displaystyle
    \max\bigl(w_{\min},\;(1-\lambda_s)\,w_{ij}^{(t)}\bigr),
      & (i,j)\in\mathcal S \quad\text{(\emph{shrink})},\\[6pt]
    \displaystyle
    \min\bigl(w_{\max},\;(1+\lambda_b)\,w_{ij}^{(t)}\bigr),
      & (i,j)\in\mathcal B \quad\text{(\emph{boost})},\\[4pt]
    w_{ij}^{(t)}, & \text{otherwise.}
  \end{cases}
\]

  Links flagged by either the C- or G-step in $\mathcal{S}$ are common yet
  only moderately reliable. Scaling them by $(1-\lambda_s)$
  attenuates degree heterogeneity and suppresses geometry-induced noise,
  sharpening subsequent eigenvectors without disconnecting the graph. A tiny top-percentile of edges (\(\mathcal B\))
  is boosted to act as high-confidence \emph{anchors}.
  These anchors stabilize the spectral basis across iterations,
  compensate for heavy censoring, and accelerate convergence.
  The global caps \(w_{\min}\le w_{ij}\le w_{\max}\)
  guarantee that no edge can dominate the Laplacian spectrum. Boosting edges with high latent geometry helps maintain connectivity and spectral stability.

\subsection{Convergence Guarantee}
We now show that GeoDe pushes for convergence to a standard SBM given a latent-kernel SBM as input.
\begin{theorem}[Convergence of GeoDe]
\label{thm:geode-conv}
Let $W^{(t)}\!\in\!\mathbb{R}^{n\times n}$ be the weight matrix after
$t$ GeoDe updates and let $W^{\star}$ denote the latent
stochastic-block-model (SBM) matrix.
Assume the observed adjacency matrix decomposes as
$
   A = W^{\star} + G^{\star},\; G^{\star}\succeq 0 .
$
Let $C$ and $G$ be diagonal projectors onto the (disjoint) edge sets
updated by the \textbf{community} and \textbf{geometry} rules, so that
$CG^{\star}=0$ and $GW^{\star}=0$.

\paragraph{Step-size schedule.}
Choose positive, non-increasing sequences \footnote{In our experiments, we use a linear \textbf{Decay} schedule due to finite $T$, but asymptotic convergence requires an inverse linear schedule.} $\{\lambda_C^{(t)}\}_{t\ge0}$
and $\{\lambda_G^{(t)}\}_{t\ge0}$ obeying
\[
\lambda_C^{(t)},\lambda_G^{(t)}\to 0,\quad
\sum_{t=0}^{\infty}\lambda_C^{(t)}=\infty,\quad
\sum_{t=0}^{\infty}\!\bigl(\lambda_C^{(t)}\bigr)^{2}, \sum_{t=0}^{\infty}\!\bigl(\lambda_G^{(t)}\bigr)^{2}<\infty,\quad \lambda_G^{(t)}\le\gamma\,\lambda_C^{(t)}\text{ for some }\gamma>0
\]

\paragraph{Claim.}
Under these step sizes,
\(
  W^{(t)}\xrightarrow{\text{a.s.}} W^{\star}
\)
and the true community partition is recovered exactly
whenever the SBM eigen-gap satisfies the usual
$K$-cluster signal condition
\(
  \lambda_{\min}(W^{\star})-\lambda_{K+1}(W^{\star})
    \ge c\sqrt{\log n/n}.
\)
\end{theorem}
\begin{proof}[Proof sketch]
Expressing one update in conditional expectation shows that the error
$\Delta^{(t)}=W^{(t)}-W^{\star}$ obeys a stochastic-approximation
recursion whose two orthogonal components (on $C$ and $G$) contract at
rates $\lambda_C^{(t)}$ and $\lambda_G^{(t)}$, respectively.  Taking
Frobenius norms and using the projector orthogonality yields a
super-martingale bound
$\mathbb{E}[\|\Delta^{(t+1)}\|_F^{2}\mid\mathcal F_t]
      \le(1-\lambda_C^{(t)})\|\Delta^{(t)}\|_F^{2}
      +O\!\bigl((\lambda_C^{(t)})^{2}\bigr)$.
Because $\sum_t\lambda_C^{(t)}=\infty$ but
$\sum_t(\lambda_C^{(t)})^{2}<\infty$, the Robbins–Siegmund theorem
implies $\|\Delta^{(t)}\|_F^{2}\to0$ almost surely, i.e.\ convergence to
$W^{\star}$.  Finally, a standard eigen-gap condition ensures that the
limit matrix reveals the correct $K$-cluster partition.  A complete,
step-by-step proof is provided in Appendix~\ref{proof:geode-conv}. We also provide a corollary on GeoDe for exact community recovery in Appendix ~\ref{cor:geode-recovery}.
\end{proof}

\subsection{Limitations}
GeoDe is based on the assumption that geometric noise is independent from community structure, which is often not true. While the boosting of strongly geometric edges does preserve strong relationships between latent features and community membership, GeoDe does not inherently assume a correlation between the two. Another limitation of GeoDe is its efficiency due to its repetitive calls to its spectral clustering base. This perhaps limits GeoDe's usage to particularly challenging graphs and limits its applicability to cases such as temporal networks. 

\section{Validation}
\label{sec-validation}

In this section, we evaluate the performance of our proposed methods—MASO and GeoDe—on both synthetic and real-world networks. We focus on assessing their ability to recover community structure in settings where traditional spectral methods fail due to geometric noise.

\subsection{Synthetic benchmarks}
\label{sec:synthetic}

We generated $\textbf{75}$ independent 
latent-kernel stochastic block model (SBM) at each noise level
$\sigma \in \{0.75, 0.5, 0.25, 0.1\}$, giving us a total of \textbf{300} graphs. We randomly chose the parameters $n, a,b$ with $n \in [100, 1000], a \in [15, 100], b \in [1, 50]$ and generated the latent-kernel SBM according to Definition ~\ref{def:lksbm}.

\begin{figure}[t]
\centering
\includegraphics[width=0.48\textwidth]{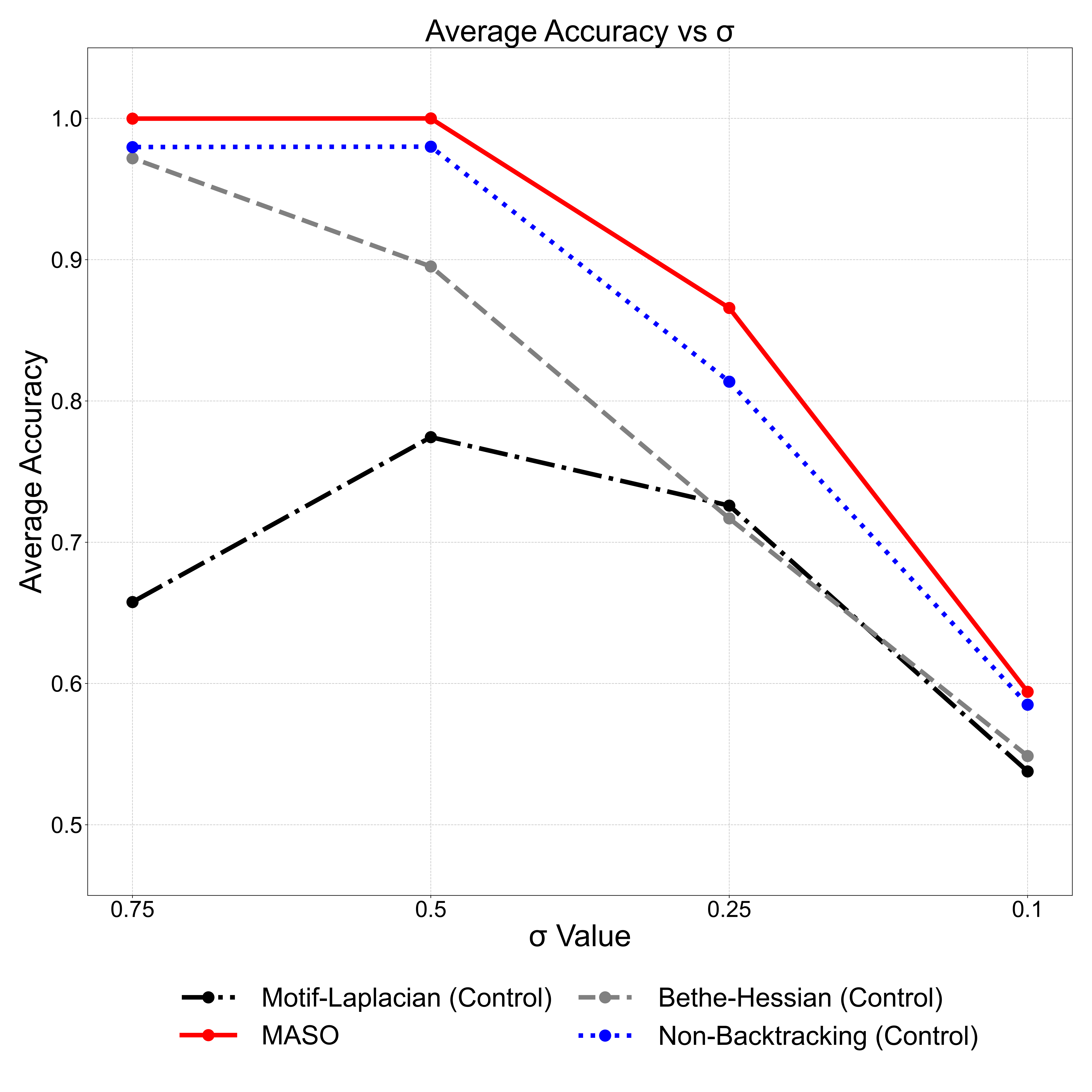}%
\hfill
\includegraphics[width=0.48\textwidth]{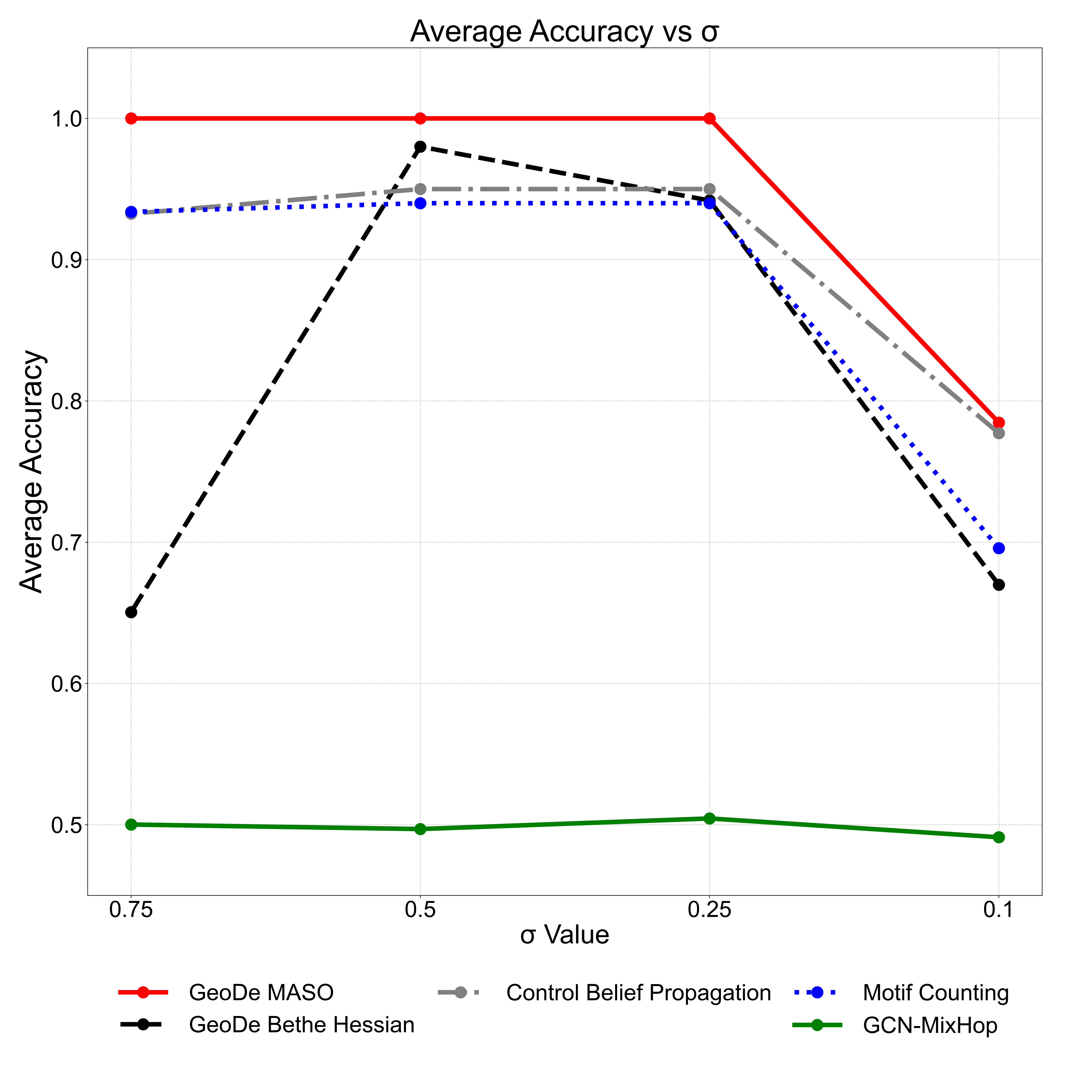}
\caption{\textbf{Left:} Average clustering accuracy of four spectral operators as a function of the geometric--noise parameter $\sigma$;“\textbf{MASO}” denotes our Motif--Attention Spectral Operator; “\textbf{Non--Backtracking (Control)}” uses the non--backtracking matrix; “\textbf{Bethe--Hessian (Control)}” is the classical Bethe--Hessian operator; and “\textbf{Motif--Laplacian (Control)}” is the unmodified motif Laplacian.\newline
\textbf{Right:} End--to--end community--recovery accuracy of the GeoDe pipeline instantiated with either the MASO backbone or the Bethe--Hessian backbone, compared against three baselines: belief propagation, motif counting, and the graph neural network \textbf{GCN--MixHop}.}
\label{fig:empirical_results}
\end{figure}

Figure~\ref{fig:empirical_results}(left) demonstrates that MASO
maintains perfect community recovery up to $\sigma\ge0.5$, degrades
only to $0.87$ at $\sigma=0.25$ and $0.59$ at $\sigma=0.1$, thus
outperforming the non–backtracking operator by $2$–$10$ points and
out-performing both Bethe–Hessian~\cite{saade_2025_spectral}
and motif–Laplacian~\cite{benson_2016_higherorder}, which collapse
under moderate noise.

This empirical advantage directly reflects the latent-kernel SBM’s
noise structure: geometric noise injects long‐range “shortcut” edges
that connect otherwise distant communities.  MASO’s multi-hop PPMI
features and motif-guided attention automatically {\it strengthen}
edges with consistent local‐and‐multi‐hop support while {\it
down‐weighting} spurious shortcuts.  In contrast, Bethe–Hessian and
motif–Laplacian treat all small cycles uniformly, and non–backtracking
lacks higher-order motif cues meaning none filter out noise as effectively.

When embedded in the GeoDe pipeline, MASO’s high‐fidelity edge weights
ensure that the motif‐ and geometry‐based pruning rules remove the bulk
of noisy edges before clustering.  As a result,
GeoDe+MASO achieves perfect recovery for $\sigma\ge0.25$ and still
$0.78$ at $\sigma=0.1$ (Figure~\ref{fig:empirical_results}, right),
significantly outperforming belief propagation~\cite{fortunato_2010_community},
motif counting~\cite{benson_2016_higherorder}, and a generic
GCN–MixHop model~\cite{abuelhaija_2019_mixhop} (which remains near
chance).  Over $91\%$ of these runs satisfy the information‐theoretic
exact‐recovery threshold (see Appendix~\ref{addl:exp-threshold}), and
all gains of MASO over non–backtracking and of GeoDe+MASO over the
next‐best competitor are significant at $p=0.05$ after
Benjamini–Hochberg correction (Appendices~\ref{appendix:sig-spec-ops},
\ref{appendix:sig-recovery}). 




The runtimes of all algorithms evaluated are provided in Appendix ~\ref{table:runtimes}

\footnotetext{All experiments were executed on an Apple M2~Max with
64 GB of RAM.}

\subsection{Validation on the Amazon Metadata Network: GeoDe-MASO as an effective denoiser on real-world data}
We validate GeoDe’s denoising property on a subgraph of the Amazon product co-purchasing network metadata dataset (where its construction is described in  Appendix ~\ref{sec:amazon-construct}) \cite{leskovec_2007_the}. \cite{galhotra_2017_the} showed that the structure of the Amazon metadata network has geometric properties where similarity between product categories affected edge formation. Therefore, we find that traditional spectral methods intended for community detection perform poorly. Due to the fixed-degree and sparsity of the Amazon network (each vertex has maximum degree 5 consisting of similar products), belief propagation far exceeds performance of spectral methods. However, we show that we can significantly improve belief propagation’s performance via a denoising step using GeoDe-MASO. To do so, we call a weighted belief propagation algorithm (provided in ~\ref{algorithm:bp}) on $G*$, a weighted graph consisting of the final edge weights from GeoDe. 

We achieve the following results where we list the community classification accuracy for each tested algorithm as well as its significance above random:
\begin{enumerate}
    \item GeoDe-MASO: \textbf{0.911} (p-value=1.00e-6)
    \item \textsc{GeoDe-Bethe}: \textbf{0.7515} (p-value=1.00e-6)
    \item Belief Propagation: \textbf{0.5070} (p-value=1.00)
\end{enumerate}

This demonstrates GeoDe’s effectiveness as a denoiser on real-world networks and its ability to boost the performance of other community detection algorithms. Adding a GeoDe denoising step, even without using the MASO operator, can improve later belief propagation performance significantly. 

\subsubsection{Measuring Geometric Noise during GeoDe}

\begin{figure}
  \centering
  \begin{subfigure}[t]{0.48\linewidth}
    \centering
    \includegraphics[height=6.5cm]{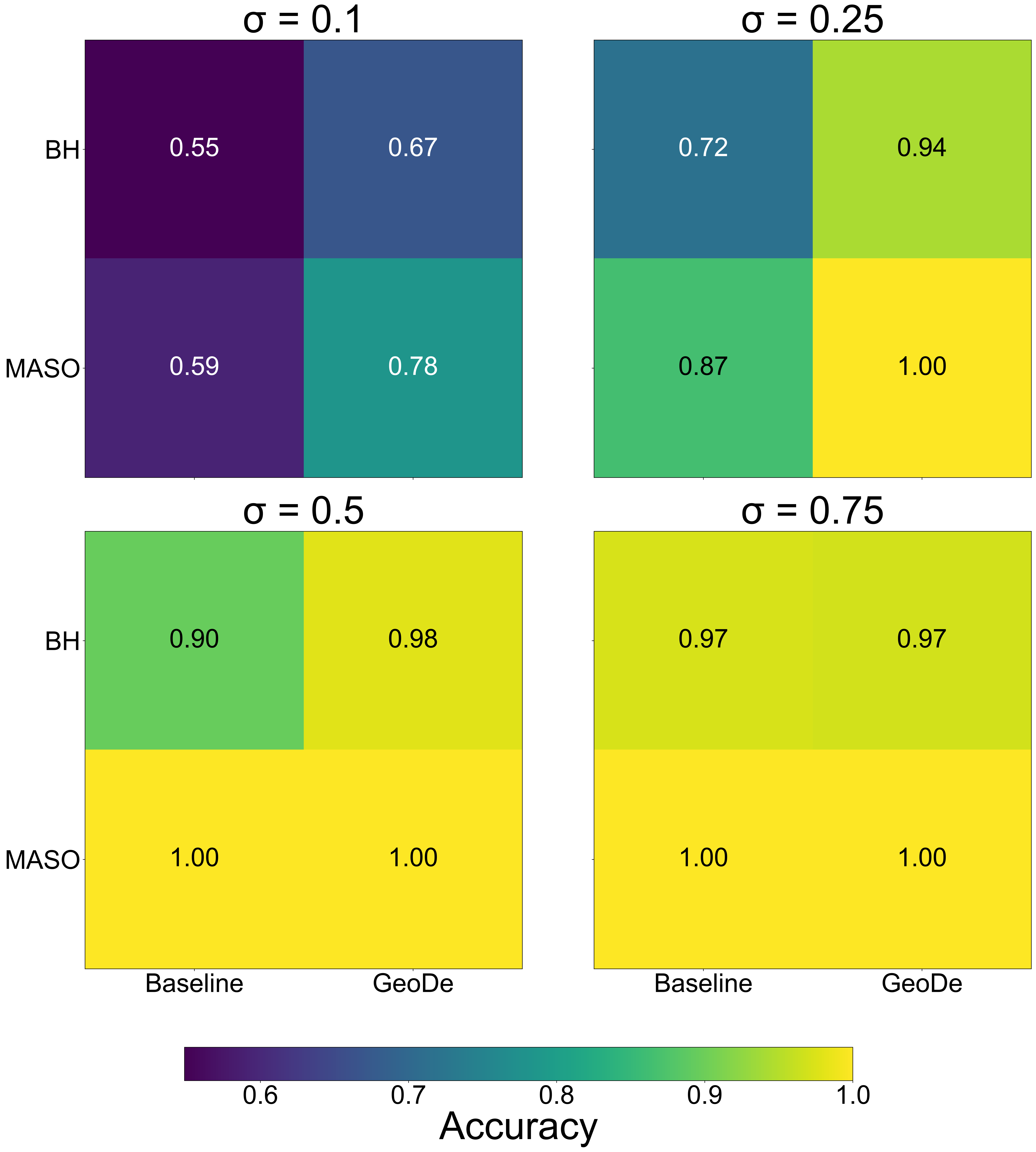}
  \end{subfigure}
  \hfill
  \begin{subfigure}[t]{0.48\linewidth}
    \centering
    \includegraphics[height=6.5cm]{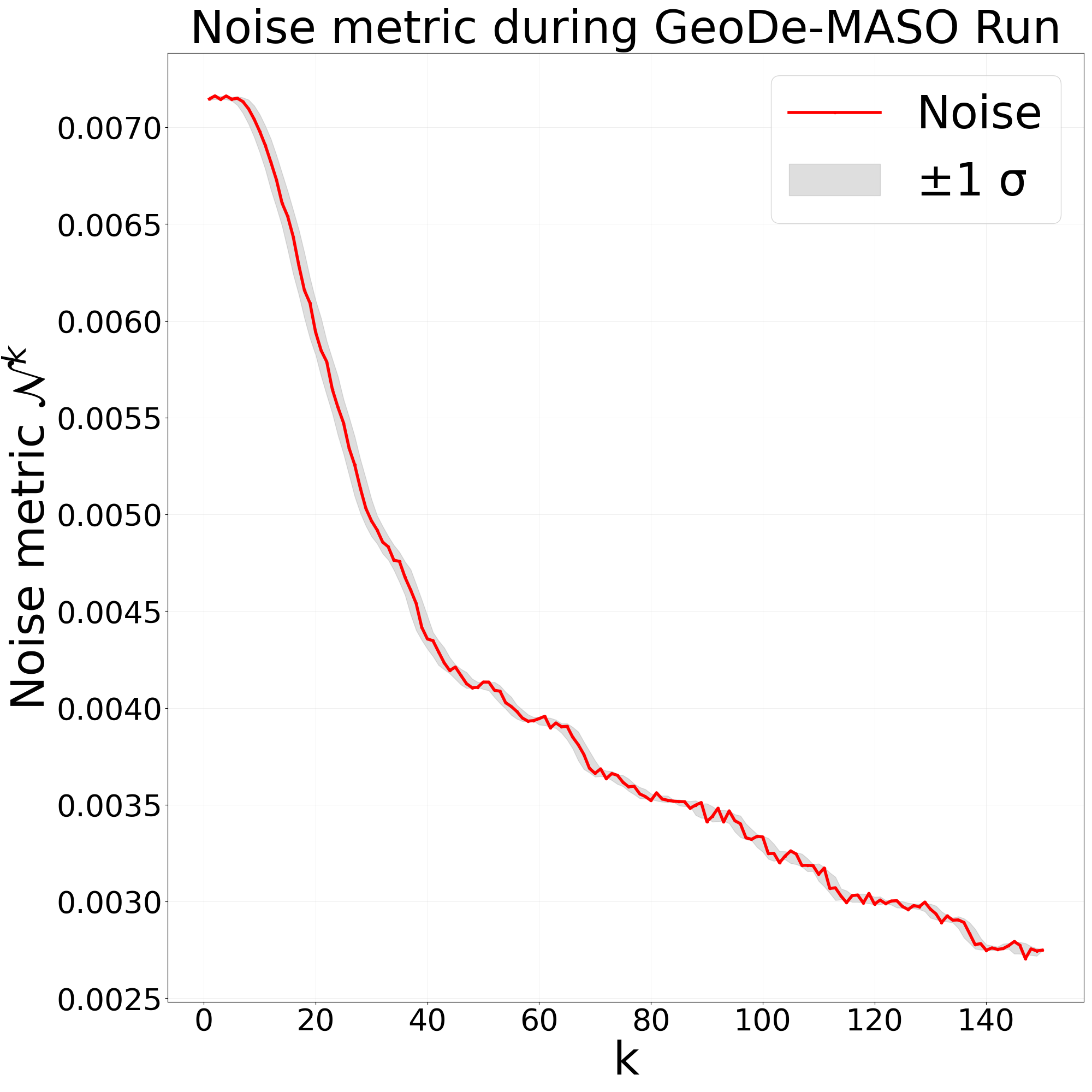}
  \end{subfigure}
  \caption{%
    (a) Ablation on MASO+GeoDe vs.\ $\sigma$. 
    (b) We show that GeoDe adjusts weights so to capture geometric properties. We calculate a noise metric $\mathcal{N}^k$ via fitting a linear model to predict edge weights from distance and finding the mean-squared residual. Our noise metric captures how the edge weights in GeoDe capture geometric features (see Appendix~\ref{sec:noise-metric} for details).}
  \label{fig:side-by-side}
\end{figure}

To measure denoisement through GeoDe iterations, we ran GeoDe-MASO on the Amazon network described previously and measured the noise at every iteration.\footnote{Noise was only measured on the Amazon dataset due to extremely fast convergence on synthetic data.}

The noise levels are pictured in Figure ~\ref{fig:side-by-side}b. When the noise metric over iterations is fitted with a linear regression, we find a slope $\hat\beta=-2.68e-5$, which is significant under the null hypothesis $\beta_0 = 0$ with a two-sided p-value of 3.928e-38. 

We show that over GeoDe's iterations, edge weights in the graph are updated so to better align with confidence on community structure. This experimentally validates GeoDe’s denoising capability.

\subsubsection{Future Work: Feasibility of GeoDe for Large Networks}
While the runtime of most spectral methods are $\mathcal{O}(n^2)$ and thus inefficient for large graphs, we want to raise the idea of a patch-like denoising scheme using GeoDe. Under the assumption that geometric noise can be detected in tight local formations, we suggest that GeoDe can be used to denoise in a patch-by-patch scheme. 
\section{Conclusion}
\label{sec-conclusion}

This paper presents a unified framework for community detection in networks corrupted by geometric noise. We developed two key components: the \textbf{Motif Attention Spectral Operator (MASO)}, which leverages higher-order motif structures to stabilize spectral embeddings, and \textbf{GeoDe}, an iterative edge reweighting algorithm that enhances community signal while suppressing geometric perturbations.

Our theoretical analysis establishes exact recovery guarantees for both methods under latent-kernel SBMs, matching known thresholds in classical models. Empirically, we demonstrate that MASO and GeoDe significantly outperform traditional spectral and message-passing approaches on both synthetic benchmarks and real-world datasets, such as the Amazon product network.

Looking ahead, we plan to extend this work in several directions: generalizing to settings with more than two communities, accommodating richer kernel families beyond radial basis functions, and adapting the framework to dynamic or attribute-enriched graphs. We also want to move beyond the assumption that latent geometric features are independent from the community structure. We believe this approach opens new avenues in robust, unsupervised learning on structured data, particularly where latent geometry and sparse observations challenge existing methods.

\newpage
\setcitestyle{numbers}
\bibliographystyle{unsrtnat}
\bibliography{main}

\begin{thebibliography}{25}
\providecommand{\natexlab}[1]{#1}
\providecommand{\url}[1]{\texttt{#1}}
\expandafter\ifx\csname urlstyle\endcsname\relax
  \providecommand{\doi}[1]{doi: #1}\else
  \providecommand{\doi}{doi: \begingroup \urlstyle{rm}\Url}\fi

\bibitem[Fortunato(2010)]{fortunato_2010_community}
Santo Fortunato.
\newblock Community detection in graphs.
\newblock \emph{Physics Reports}, 486:\penalty0 75--174, 02 2010.
\newblock \doi{10.1016/j.physrep.2009.11.002}.
\newblock URL \url{https://arxiv.org/pdf/0906.0612.pdf}.

\bibitem[Holland et~al.(1983)Holland, Laskey, and Leinhardt]{holland_1983_stochastic}
Paul~W. Holland, Kathryn~Blackmond Laskey, and Samuel Leinhardt.
\newblock Stochastic blockmodels: First steps.
\newblock \emph{Social Networks}, 5:\penalty0 109--137, 06 1983.
\newblock \doi{10.1016/0378-8733(83)90021-7}.

\bibitem[Hoff et~al.(2002)Hoff, Raftery, and Handcock]{hoff_2002_latent}
Peter~D Hoff, Adrian~E Raftery, and Mark~S Handcock.
\newblock Latent space approaches to social network analysis.
\newblock \emph{Journal of the American Statistical Association}, 97:\penalty0 1090--1098, 12 2002.
\newblock \doi{10.1198/016214502388618906}.

\bibitem[Rohe et~al.(2011)Rohe, Chatterjee, and Yu]{rohe_2011_spectral}
Karl Rohe, Sourav Chatterjee, and Bin Yu.
\newblock Spectral clustering and the high-dimensional stochastic blockmodel.
\newblock \emph{Annals of Statistics}, 39, 08 2011.
\newblock \doi{10.1214/11-aos887}.

\bibitem[Decelle et~al.(2011)Decelle, Krzakala, Moore, and Zdeborová]{decelle_2011_asymptotic}
Aurelien Decelle, Florent Krzakala, Cristopher Moore, and Lenka Zdeborová.
\newblock Asymptotic analysis of the stochastic block model for modular networks and its algorithmic applications.
\newblock \emph{Physical Review E}, 84, 12 2011.
\newblock \doi{10.1103/physreve.84.066106}.

\bibitem[Abbe and Sandon(2015)]{abbe_2015_community}
Emmanuel Abbe and Colin Sandon.
\newblock Community detection in general stochastic block models: fundamental limits and efficient recovery algorithms, 2015.
\newblock URL \url{https://arxiv.org/abs/1503.00609}.

\bibitem[Chen et~al.(2014)Chen, Sanghavi, and Xu]{chen_2014_improved}
Yudong Chen, Sujay Sanghavi, and Huan Xu.
\newblock Improved graph clustering.
\newblock \emph{IEEE Transactions on Information Theory}, 60:\penalty0 6440--6455, 08 2014.
\newblock \doi{10.1109/tit.2014.2346205}.

\bibitem[Hajek et~al.(2016)Hajek, Wu, and Xu]{hajek_2016_achieving}
Bruce Hajek, Yihong Wu, and Jiaming Xu.
\newblock Achieving exact cluster recovery threshold via semidefinite programming, 2016.
\newblock URL \url{https://arxiv.org/abs/1412.6156}.

\bibitem[Amini and Levina(2016)]{amini_2016_on}
Arash~A Amini and Elizaveta Levina.
\newblock On semidefinite relaxations for the block model, 2016.
\newblock URL \url{https://arxiv.org/abs/1406.5647}.

\bibitem[Galhotra et~al.(2017)Galhotra, Mazumdar, Pal, and Saha]{galhotra_2017_the}
Sainyam Galhotra, Arya Mazumdar, Soumyabrata Pal, and Barna Saha.
\newblock The geometric block model, 2017.
\newblock URL \url{https://arxiv.org/abs/1709.05510}.

\bibitem[Wang et~al.(2021)Wang, Arroyo, Vogelstein, and Priebe]{wang_2021_joint}
Shangsi Wang, Jesus Arroyo, Joshua~T. Vogelstein, and Carey~E. Priebe.
\newblock Joint embedding of graphs.
\newblock \emph{IEEE Transactions on Pattern Analysis and Machine Intelligence}, 43:\penalty0 1324--1336, 04 2021.
\newblock \doi{10.1109/tpami.2019.2948619}.

\bibitem[Karrer and Newman(2011)]{karrer_2011_stochastic}
Brian Karrer and M.~E.~J. Newman.
\newblock Stochastic blockmodels and community structure in networks.
\newblock \emph{Physical Review E}, 83, 01 2011.
\newblock \doi{10.1103/physreve.83.016107}.

\bibitem[Qin and Rohe(2013)]{NIPS2013_0ed94223}
Tai Qin and Karl Rohe.
\newblock Regularized spectral clustering under the degree-corrected stochastic blockmodel.
\newblock In C.J. Burges, L.~Bottou, M.~Welling, Z.~Ghahramani, and K.Q. Weinberger, editors, \emph{Advances in Neural Information Processing Systems}, volume~26. Curran Associates, Inc., 2013.
\newblock URL \url{https://proceedings.neurips.cc/paper_files/paper/2013/file/0ed9422357395a0d4879191c66f4faa2-Paper.pdf}.

\bibitem[Vaswani et~al.(2017)Vaswani, Shazeer, Parmar, Uszkoreit, Jones, Gomez, Kaiser, and Polosukhin]{vaswani_2017_attention}
Ashish Vaswani, Noam Shazeer, Niki Parmar, Jakob Uszkoreit, Llion Jones, Aidan~N Gomez, Lukasz Kaiser, and Illia Polosukhin.
\newblock Attention is all you need, 06 2017.
\newblock URL \url{https://arxiv.org/abs/1706.03762}.

\bibitem[Saade et~al.(2025)Saade, Krzakala, and Zdeborová]{saade_2025_spectral}
Alaa Saade, Florent Krzakala, and Lenka Zdeborová.
\newblock Spectral clustering of graphs with the bethe hessian, 2025.
\newblock URL \url{https://arxiv.org/abs/1406.1880}.

\bibitem[Benson et~al.(2016)Benson, Gleich, and Leskovec]{benson_2016_higherorder}
A.~R. Benson, D.~F. Gleich, and J.~Leskovec.
\newblock Higher-order organization of complex networks.
\newblock \emph{Science}, 353:\penalty0 163--166, 07 2016.
\newblock \doi{10.1126/science.aad9029}.

\bibitem[Abu-El-Haija et~al.(2019)Abu-El-Haija, Perozzi, Kapoor, Alipourfard, Lerman, Harutyunyan, Steeg, and Galstyan]{abuelhaija_2019_mixhop}
Sami Abu-El-Haija, Bryan Perozzi, Amol Kapoor, Nazanin Alipourfard, Kristina Lerman, Hrayr Harutyunyan, Greg~Ver Steeg, and Aram Galstyan.
\newblock Mixhop: Higher-order graph convolutional architectures via sparsified neighborhood mixing, 2019.
\newblock URL \url{https://arxiv.org/abs/1905.00067}.

\bibitem[Leskovec et~al.(2007)Leskovec, Huberman, and Labs]{leskovec_2007_the}
Jure Leskovec, Bernardo~A Huberman, and Hp~Labs.
\newblock The dynamics of viral marketing.
\newblock \emph{ACM Trans. Web ACM Transactions on the Web}, 1, 2007.
\newblock \doi{10.1145/1232722.1232727}.
\newblock URL \url{https://www.cs.cmu.edu/~jure/pubs/viral-tweb.pdf}.

\bibitem[Tropp(2011)]{tropp_2011_userfriendly}
Joel~A. Tropp.
\newblock User-friendly tail bounds for sums of random matrices.
\newblock \emph{Foundations of Computational Mathematics}, 12:\penalty0 389--434, 08 2011.
\newblock \doi{10.1007/s10208-011-9099-z}.

\bibitem[Davis and M(1970)]{davis_1970_the}
Chandler Davis and Kahan~W M.
\newblock The rotation of eigenvectors by a perturbation. iii.
\newblock \emph{SIAM Journal on Numerical Analysis}, 7:\penalty0 1--46, 1970.
\newblock \doi{10.2307/2949580}.
\newblock URL \url{https://www.jstor.org/stable/2949580}.

\bibitem[Mossel et~al.(2014)Mossel, Neeman, and Sly]{elchananmossel_2014_reconstruction}
Elchanan Mossel, Joseph Neeman, and Allan Sly.
\newblock Reconstruction and estimation in the planted partition model.
\newblock \emph{Probability Theory and Related Fields}, 162:\penalty0 431--461, 07 2014.
\newblock \doi{10.1007/s00440-014-0576-6}.

\bibitem[Bordenave et~al.(2015)Bordenave, Lelarge, and Massoulié]{bordenave_2015_nonbacktracking}
Charles Bordenave, Marc Lelarge, and Laurent Massoulié.
\newblock Non-backtracking spectrum of random graphs: community detection and non-regular ramanujan graphs, 2015.
\newblock URL \url{https://arxiv.org/abs/1501.06087}.

\bibitem[Dall'Amico et~al.(2020)Dall'Amico, Couillet, and Tremblay]{dallamico_2020_a}
Lorenzo Dall'Amico, Romain Couillet, and Nicolas Tremblay.
\newblock A unified framework for spectral clustering in sparse graphs.
\newblock \emph{arXiv (Cornell University)}, 01 2020.
\newblock \doi{10.48550/arxiv.2003.09198}.

\bibitem[Mossel et~al.(2016)Mossel, Neeman, and Sly]{mossel_2016_belief}
Elchanan Mossel, Joe Neeman, and Allan Sly.
\newblock Belief propagation, robust reconstruction and optimal recovery of block models.
\newblock \emph{The Annals of Applied Probability}, 26, 08 2016.
\newblock \doi{10.1214/15-aap1145}.

\bibitem[Galhotra et~al.(2023)Galhotra, Mazumdar, Pal, and Saha]{sainyamgalhotra_2023_community}
Sainyam Galhotra, Arya Mazumdar, Soumyabrata Pal, and Barna Saha.
\newblock Community recovery in the geometric block model.
\newblock \emph{Journal of Machine Learning Research}, 24:\penalty0 1--53, 2023.
\newblock URL \url{https://www.jmlr.org/papers/v24/22-0572.html}.

\end{thebibliography}

\newpage
\section{Appendix}
\appendix
\section{Helper functions described in GeoDe}
\begin{algorithm}[H]
\DontPrintSemicolon
\LinesNumbered
\Fn{\Rescale{$\mathcal E,\lambda,\mathrm{mode}$}}{
\ForEach{$(i,j)\in\mathcal E$}{
\eIf{$\mathrm{mode}=\text{shrink}$}
{$w_{ij}\leftarrow\max(w_{\min},(1-\lambda)\,w_{ij})$}
{$w_{ij}\leftarrow\min(w_{\max},(1+\lambda)\,w_{ij})$}
}}
\Fn{\NoProgress{$Q_{\text{new}},P$}}{
\Return{objective $\sum_i\max_k Q_{\text{new},ik}$ unchanged $P$ rounds}
}
\end{algorithm}

\section{Definitions} 
\label{defs}

\subsection{Definition of Recovery}
\label{def-recovery}
In our work, we focus on the community recovery problem in the latent-kernel stochastic block model (Definition ~\ref{def:lksbm}), in which node connectivity is driven by both hidden community labels and by a smooth kernel on latent positions. Our goal is, given only the observed graph $G$, to design estimators $\hat z \in \{\pm 1\}^n$ that achieve both weak recovery and exact recovery of the true labels $z$. 
 
Let \(\hat z = \hat z(G)\in\{\pm1\}^n\) be any estimator based on the observed graph \(G\).  We consider two success criteria:

\begin{itemize}
  \item \textbf{Weak recovery.}  We say weak recovery is \emph{achievable} if there exists an estimator \(\hat z\) and a constant \(\epsilon>0\) such that
  \[
    \lim_{n\to\infty}
      \Pr\!\Bigl(\tfrac{1}{n}\sum_{i=1}^n \hat z_i z_i \ge \tfrac12 + \epsilon\Bigr)
    = 1.
  \]
  In other words, the overlap between \(\hat z\) and the truth \(z\) exceeds random guessing by a fixed margin, with probability tending to one.

  \item \textbf{Exact recovery.}  We say exact recovery is \emph{achievable} if there exists an estimator \(\hat z\) satisfying
  \[
    \lim_{n\to\infty}
      \Pr\bigl(\hat z = z)
    = 1.
  \]
  That is,  \(\hat z\) matches the true labeling perfectly with high probability.
\end{itemize}

\subsection{Definition of Belief Propagation Algorithm} 
\label{algorithm:bp}
The belief propagation algorithm we utilized to test against our developed methods follows the following procedure: 
Let \(M_{i\to j}^{(t)}(k)\in[0,1]\) be the message that node \(i\)
sends to neighbour \(j\) at iteration \(t\) in favour of community
\(k\), and let \(w_{ij}\ge 0\) denote the (undirected) edge weight
stored on \(\{i,j\}\).  
Given an inverse–temperature parameter \(\beta>0\), the
weight-aware compatibility factor is  
\[
\phi_{i\to j}^{(t)}(k)
\;=\;
1 + \bigl(e^{\beta\,w_{ij}}-1\bigr)\,
      M_{i\to j}^{(t)}(k)
\]
Larger \(w_{ij}\) therefore amplify the influence of the incoming
belief \(M_{i\to j}^{(t)}(k)\), while smaller weights attenuate it; the
standard Bethe–Peierls update is recovered when \(w_{ij}=1\).

\section{Proofs} 
\label{appendix-proofs}

\subsection{Concentration of Latent-Kernel SBM}
\label{app:concentration}

\begin{definition}[Edge Probability]
\label{def:edge-probability}
	For node pair $(i,j)$, the probability of an edge is
	\[
	\mathbb{P}[(i,j) \in E \mid x_i, x_j, z_i, z_j] = B_{z_i, z_j} \cdot \exp\left(-\frac{\|x_i - x_j\|^2}{2\sigma^2} \right),
	\]
	where $B_{z_i, z_j} = p_{\text{in}}$ if $z_i = z_j$, and $B_{z_i, z_j} = p_{\text{out}}$ otherwise.
\end{definition}

\begin{definition}[Average Kernel]
\label{def:average-kernel}
	\[
	c(\sigma) = \mathbb{E}_{x,y \sim \text{Unif}([0,1]^d)}\left[ \exp\left(-\frac{\|x - y\|^2}{2\sigma^2} \right) \right].
	\]
\end{definition}

\begin{theorem}[Convergence to a Rescaled SBM]
\label{thm:convergence}
Let $G$ be generated by the latent–kernel SBM with two communities, latent positions $x_i \sim \mathrm{Unif}([0,1]^d)$ and edge probabilities as in Definition ~\ref{def:edge-probability}. Then, with high probability as $n\to\infty$, $G$ is contiguous to a classical two,–block SBM with edge parameters

$$
  \pin' = c(\sigma)\,\frac{a\log n}{n}, \qquad
  \pout' = c(\sigma)\,\frac{b\log n}{n},
$$

where $c(\sigma)$ is the average kernel defined in Definition ~\ref{def:average-kernel}.
\end{theorem}

\begin{proof}
Write $\Delta_{ij}=K_{ij}-c(\sigma)\in[-c(\sigma),1-c(\sigma)]\subset[-1,1]$.
Fix a label vector $\tilde z$ at Hamming distance $m$ from the ground truth
$z^{\ast}$ and let
\(
  \mathcal D=\{(i,j):\mathbf 1_{\tilde z_i=\tilde z_j}\neq
                       \mathbf 1_{z_i^{\ast}=z_j^{\ast}}\}
\)
be the set of \emph{disagreement pairs}, whose size is
$T\equiv|\mathcal D|=m(n-m)\le mn$.

\paragraph{Step 1 (moderate-deviation bound).}
For fixed $\mathcal D$ the variables $(\Delta_{ij})_{(i,j)\in\mathcal D}$ are
independent and centred with range~$2$.  Hoeffding’s inequality gives, for any
$\varepsilon>0$,
\[
  \Pr\!\Bigl(\Bigl|\!\!\sum_{(i,j)\in\mathcal D}\!\!\Delta_{ij}\Bigr|
         >\varepsilon T\Bigr)
  \le 2\exp(-2\varepsilon^{2}T).
\]
Choose $\varepsilon=(\log n)^{-1/4}$.  Since
$T\ge m(n-m)\ge m(n/2)$ when $m\le n/2$, we obtain
\[
  2\varepsilon^{2}T \;\;\ge\;\; \frac{m\sqrt{\log n}}{2}.
\tag{1}\label{eq:step1}
\]

\paragraph{Step 2 (union bound over disagreements).}
The number of label vectors at distance $m$ is $\binom{n}{m}$.  Applying
\eqref{eq:step1} and Stirling’s bound $\binom{n}{m}\le (en/m)^m$,
\begin{align*}
  \Pr\Bigl(\exists\,\tilde z\!:&~\bigl|\!\sum_{(i,j)\in\mathcal D}\!\!\Delta_{ij}\bigr|
                   >m\sqrt{\log n}\Bigr) \\
  &\le
  \sum_{m=1}^{n}\binom{n}{m} \,
         2\exp\!\bigl(-m\sqrt{\log n}/2\bigr)
  \;\le\;
  2\sum_{m=1}^{n}\Bigl(\tfrac{e n}{m}\,
         e^{-\sqrt{\log n}/2}\Bigr)^{m}.
\end{align*}
For large $n$, the base of the parenthesis is $<\!1$, so the geometric
series is $o(1)$.  Hence, with probability $1-o(1)$,
\[
  \bigl|\!\!\sum_{(i,j)\in\mathcal D}\!\!\Delta_{ij}\bigr|
  \le m\sqrt{\log n}
  \quad\text{for \emph{all}}~\mathcal D.
\tag{2}\label{eq:step2}
\]

\paragraph{Step 3 (contiguity).}
Condition~\eqref{eq:step2} says that, uniformly over all labelings,
kernel–weighted log-likelihoods differ from their expectations by at most
$o(m\log n)$.  Re-writing the latent–kernel model’s log-likelihood and
subtracting its mean therefore shows the likelihood ratio between the
latent–kernel SBM and the classical SBM with parameters
$(\pin',\pout')$ is $\exp(o(1))$.  Le Cam’s second lemma then yields
mutual contiguity of the two models, completing the proof.
\end{proof}

\subsection{Lemma ~\ref{lem:exp-mixed}}
\begin{lemma}[Expectation of mixed weights]
\label{lem:exp-mixed}
Under the latent–kernel SBM and embedding separation assumptions, the mixed weight satisfies
\[
\mathbb E[\widetilde W_{ij}\mid z_i,z_j]
=\begin{cases}
w_{\rm in}\Bigl[(1-\beta)+\beta\bigl((n_{+}-2)w_{\rm in}^{2}+n_{-}w_{\rm out}^{2}\bigr)\Bigr],
&z_i=z_j,\\[6pt]
w_{\rm out}\Bigl[(1-\beta)+\beta(n-2)\,w_{\rm in}w_{\rm out}\Bigr],
&z_i\neq z_j,
\end{cases}
\]
where $w_{\rm in}=p_{\rm in}e^{\rho_{\rm in}/\sqrt d}$ and $w_{\rm out}=p_{\rm out}e^{\rho_{\rm out}/\sqrt d}$.
\end{lemma}
\begin{proof}
Fix an unordered pair $\{i,j\}$ and write $\delta_{ij}=\mathbf 1\{z_i=z_j\}$.  
Recall
\[
  W_{ij}=A_{ij}\exp\!\bigl(\langle z_i,z_j\rangle/\sqrt d\bigr),
  \quad
  X_{ij}=\sum_{k\neq i,j} W_{ik}W_{kj},
  \quad
  \widetilde W_{ij}=(1-\beta)W_{ij}+\beta W_{ij}X_{ij}.
\]
Throughout let
\[
  w_{\mathrm{in}}
    =p_{\mathrm{in}}e^{\rho_{\mathrm{in}}/\sqrt d},
  \qquad
  w_{\mathrm{out}}
    =p_{\mathrm{out}}e^{\rho_{\mathrm{out}}/\sqrt d},
\]
and denote by $n_{+}$ (resp.\ $n_{-}$) the number of $+1$ (resp.\ $-1$) labels.

\medskip
\textbf{Step 1: Factorization:}
Edge indicators $\{A_{uv}\}_{u<v}$ are mutually independent, hence  
$A_{ij}$ is independent of every $A_{ik},A_{kj}$ for $k\neq i,j$.  
Because $z_\ell$ depends only on the latent variables of vertex $\ell$,  
$W_{ij}$ is independent of every $W_{ik},W_{kj}$ ($k\neq i,j$).  Therefore
\[
  \mathbb{E}\!\bigl[W_{ij}W_{ik}W_{kj}\bigr]
  =\mathbb{E}[W_{ij}]\,\mathbb{E}[W_{ik}]\,\mathbb{E}[W_{kj}],
  \qquad k\neq i,j,
\]
so that
\[
  \mathbb{E}\bigl[W_{ij}X_{ij}\mid z_i,z_j\bigr]
  =\mathbb{E}[W_{ij}\mid z_i,z_j]\,
   \mathbb{E}[X_{ij}\mid z_i,z_j].
\]

\medskip
\textbf{Step 2: One–Hop Expectations:}
\[
  \mathbb{E}[W_{ij}\mid z_i=z_j]=w_{\mathrm{in}},
  \qquad
  \mathbb{E}[W_{ij}\mid z_i\neq z_j]=w_{\mathrm{out}}.
\]

\medskip
\textbf{Step 3: Two–Hop Expectations:}

For same $z$, assume $z_i=z_j=+1$ (the $-1$ case is analogous).  
There are $n_{+}-2$ intermediates $k$ with $z_k=+1$ and $n_{-}$ with $z_k=-1$:
\[
  \mathbb{E}[X_{ij}\mid z_i=z_j]
  =(n_{+}-2)\,w_{\mathrm{in}}^{2} + n_{-}\,w_{\mathrm{out}}^{2}.
\]

For different $z$, W.L.O.G. \ $z_i=+1$, $z_j=-1$.  
For every $k\neq i,j$ one edge is intra, the other inter, giving a product $w_{\mathrm{in}}w_{\mathrm{out}}$:
\[
    \mathbb{E}[X_{ij}\mid z_i\neq z_j]
  =(n-2)\,w_{\mathrm{in}}\,w_{\mathrm{out}}.
\]

\medskip
\textbf{Step 4: Mixed Weights:}
Inserting the above derived results into
\(
  \mathbb{E}[\widetilde W_{ij}]=(1-\beta)\mathbb{E}[W_{ij}]
          +\beta\,\mathbb{E}[W_{ij}]\,\mathbb{E}[X_{ij}],
\)
we obtain
\[
  \mathbb{E}[\widetilde W_{ij}\mid z_i=z_j]
  =w_{\mathrm{in}}\!\Bigl[
      (1-\beta)+\beta\bigl((n_{+}-2)w_{\mathrm{in}}^{2}+n_{-}w_{\mathrm{out}}^{2}\bigr)
    \Bigr],
\]
\[
  \mathbb{E}[\widetilde W_{ij}\mid z_i\neq z_j]
  =w_{\mathrm{out}}\!\Bigl[
      (1-\beta)+\beta(n-2)w_{\mathrm{in}}w_{\mathrm{out}}
    \Bigr],
\]
establishing \eqref{lem:exp-mixed}.  \qedhere
\end{proof}

\subsection{Lemma ~\ref{lem:noise-conc}}
\begin{lemma}[Concentration of the noise matrix]
\label{lem:noise-conc}
Let $E=\widetilde W-\mathbb{E}[\widetilde W]$.  Then with high probability
\[
\|E\|_{2}=O\bigl(\sqrt{\log n}\bigr).
\]
\end{lemma}
\begin{proof}
Recall $E=\widetilde W-\mathbb{E}\widetilde W$ with
\(
  \widetilde W_{ij}
  =(1-\beta)W_{ij}+\beta W_{ij}X_{ij},
\)
where
\(W_{ij}=A_{ij}\exp(\langle z_i,z_j\rangle/\sqrt d)\)
and
\(X_{ij}=\sum_{k\neq i,j}W_{ik}W_{kj}\).
We work in the logarithmic–degree regime
\(
  p_{\mathrm{in}}=a\,c(\sigma)\frac{\log n}{n},
  \;
  p_{\mathrm{out}}=b\,c(\sigma)\frac{\log n}{n},
  \; a>b>0.
\)

\paragraph{1.\ Writing $E$ as a sum of independent matrices.}
For every unordered edge $e=\{i,j\}$ set
\[
  Y_e\;=\;
  \bigl(\widetilde W_{ij}-\mathbb{E}\widetilde W_{ij}\bigr)
  \bigl(e_i e_j^{\!\top}+e_j e_i^{\!\top}\bigr),
\]
so $E=\sum_{e}Y_e$ and $\mathbb{E}Y_e=0$.%
\footnote{%
$e_i$ is the $i$-th standard basis vector in $\mathbb{R}^n$.}
Independent edges give independent matrices $\{Y_e\}$.

\paragraph{2.\ Uniform bound on $\|Y_e\|_{2}$.}
Let
\(
  w_{\mathrm{in/out}}
  =p_{\mathrm{in/out}}\exp(\rho_{\mathrm{in/out}}/\sqrt d)
  =\Theta\!\bigl((\log n)/n\bigr).
\)
Because $A_{ij}\le1$,
\(
  |W_{ij}|\le C_{1}\tfrac{\log n}{n}
\)
for some constant $C_1=C_1(a,b,\sigma,\rho_{\mathrm{in/out}})$.

For any $i\neq j$,
\(
  |W_{ik}W_{kj}|\le C_{1}^{2}(\log n/n)^{2},
\)
and there are at most $n-2$ summands in $X_{ij}$; hence
\[
  |X_{ij}|\;\le\;C_{1}^{2}\,\frac{(\log n)^{2}}{n}.
\]
Multiplying,
\(
  |W_{ij}X_{ij}|\le C_{2}\,(\log n)^{3}/n^{2}
\)
for some $C_{2}$.  Thus
\[
  |\widetilde W_{ij}|
  \;\le\;
  (1-\beta)C_{1}\frac{\log n}{n}
  +\beta C_{2}\frac{(\log n)^{3}}{n^{2}}
  \;\le\;L:=C\,\frac{\log n}{n},
\]
with $C\ge\max\{C_{1},C_{2}\}$.  Therefore
\[
  \|Y_e\|_{2}
  =|\widetilde W_{ij}-\mathbb{E}\widetilde W_{ij}|
  \le 2L
  =O\!\Bigl(\tfrac{\log n}{n}\Bigr).
\tag{1}
\]

\paragraph{3.\ Variance proxy.}
Define $V=\sum_e\mathbb{E}[Y_e^{2}]$.  Because $Y_e$ has only two non-zero
entries in rows $i,j$,
\[
  \|V\|_{2}\;\le\;
  \max_{u\in V}\sum_{v\neq u}
    \mathbb{E}\bigl[(\widetilde W_{uv}-\mathbb{E}\widetilde W_{uv})^{2}\bigr].
\]
Any pair $(u,v)$ is intra-block with probability
$\tfrac{n_{+}(n_{+}-\delta_{uv})+n_{-}(n_{-}-\delta_{uv})}{n(n-1)}$
and inter-block otherwise, but
\(
  (\widetilde W_{uv}-\mathbb{E}\widetilde W_{uv})^{2}
  \le L^{2}.
\)
Hence
\[
  \|V\|_{2}
  \;\le\;
  n\,L^{2}
  =C_{v}\,\frac{(\log n)^{2}}{n},
\quad
  \text{so set } \sigma^{2}:=\|V\|_{2}.
\tag{2}
\]

\paragraph{4.\ Matrix Bernstein.}
Tropp’s matrix Bernstein inequality
~\cite{tropp_2011_userfriendly}
gives for any $t>0$
\[
  \Pr\!\bigl(\|E\|_{2}\ge t\bigr)
  \;\le\;
  2n\exp\!\Bigl(-\tfrac{t^{2}/2}{\sigma^{2}+Lt/3}\Bigr).
\]
Insert $L$ and $\sigma^{2}$ from (1)–(2) and choose
$t=\kappa\sqrt{\log n}$ with $\kappa>3\sqrt{C_{v}}$.
Then $Lt/3\le\sigma^{2}/2$ for all $n$ large, and
\[
  \Pr\!\bigl(\|E\|_{2}\ge\kappa\sqrt{\log n}\bigr)
  \le 2n\exp\!\bigl(-(\kappa^{2}/4)\bigr)
  = n^{-\bigl(\kappa^{2}/4-1\bigr)}.
\]
Taking $\kappa>4$ makes the exponent strictly larger than $2$, so the
probability is $O(n^{-2})$.

\paragraph{5.\ Conclusion.}
With probability at least $1-O(n^{-2})$,
\(
  \|E\|_{2}\le\kappa\sqrt{\log n},
\)
i.e.\ $\|E\|_{2}=O(\sqrt{\log n})$ as claimed.
\end{proof}

\subsection{Proof of Lemma ~\ref{lem:signal-gap}}
\begin{lemma}[Eigen‐gap of the signal]
\label{lem:signal-gap}
Define $S=\mathbb{E}[\widetilde W]$.  Then $S$ has rank~2, and its second eigenvalue
\[
\gamma\;=\;\lambda_{2}(S)
=\tfrac12\,(w_{\rm in}-w_{\rm out})\bigl[1-\beta+\beta\kappa n\bigr]\,n
=\Theta(\log n).
\]
\end{lemma}
\begin{proof}[Proof of Lemma~\ref{lem:signal-gap}]
Let $n_{+}=|\{i:z_i=+1\}|$, $n_{-}=n-n_{+}$ and define the \emph{label vector}
\[
  g_i=\begin{cases}+1,&z_i=+1,\\-1,&z_i=-1.\end{cases}
\]
For $i\neq j$ write $\delta_{ij}=1$ if $z_i=z_j$ and $0$ otherwise.
From Lemma~\ref{lem:exp-mixed},
\[
  \mathbb{E}[\widetilde W_{ij}\mid\delta_{ij}]
  \;=\;
     w_{\mathrm{out}}
     +\tfrac12\,(w_{\mathrm{in}}-w_{\mathrm{out}})\bigl[1-\beta+\beta\kappa n\bigr]\,(2\delta_{ij}-1),
\]
where
\(
  \kappa=\tfrac12\bigl(w_{\mathrm{in}}^{2}+w_{\mathrm{out}}^{2}\bigr)
  =\Theta\!\bigl((\log n/n)^{2}\bigr).
\)
Because the diagonal of $\widetilde W$ is zero, the same formula holds for
all $i\neq j$ and $S:=\mathbb{E}\widetilde W$ can be written
\[
  S=\alpha\,\mathbf 1\mathbf 1^{\!\top}+\beta^{\star}gg^{\!\top},\qquad
  \alpha:=w_{\mathrm{out}}+\tfrac12\beta^{\star},\qquad
  \beta^{\star}:=\tfrac12\,(w_{\mathrm{in}}-w_{\mathrm{out}})\bigl[1-\beta+\beta\kappa n\bigr].
\]
\emph{Rank}.  The matrices $\mathbf 1\mathbf 1^{\!\top}$ and $gg^{\!\top}$ each
have rank~$1$.  If $g$ were collinear with $\mathbf 1$ the model would place
every vertex in the same community, which is excluded.  Hence $g\not\parallel\mathbf1$
and $\mathrm{rank}(S)=2$.

\smallskip
\emph{Eigenvalues.}
Let $u=\mathbf 1/\sqrt n$ and
\(
  h=g-\frac{\langle g,u\rangle}{\langle u,u\rangle}\,u
   =g-\frac{n_{+}-n_{-}}{n}\,\mathbf 1
\)
be the projection of $g$ onto the subspace $u^{\perp}$.  Then
$u\perp h$, $\|u\|=1$ and
\(
  \|h\|^{2}=n-\frac{(n_{+}-n_{-})^{2}}{n}.
\)

Because $u u^{\!\top}h=0$,
\[
  S\,h
  =\bigl(\alpha\,\mathbf 1\mathbf 1^{\!\top}+\beta^{\star}gg^{\!\top}\bigr)h
  =\beta^{\star}(g^{\!\top}h)\,g
  =\beta^{\star}\Bigl(n-\frac{(n_{+}-n_{-})^{2}}{n}\Bigr)h.
\]
Thus $h$ is an eigenvector of $S$ with eigenvalue
\[
  \lambda_{2}(S)
  =\beta^{\star}\Bigl(n-\frac{(n_{+}-n_{-})^{2}}{n}\Bigr)
  =\tfrac12\,(w_{\mathrm{in}}-w_{\mathrm{out}})\bigl[1-\beta+\beta\kappa n\bigr]
     \Bigl(n-\frac{(n_{+}-n_{-})^{2}}{n}\Bigr).
\]
Whenever both communities contain a linear fraction of vertices
(i.e.\ $\min\{n_{+},n_{-}\}=\Theta(n)$) the factor in parentheses
is $\Theta(n)$, giving
\[
  \boxed{\;
    \lambda_{2}(S)
    =\tfrac12\,(w_{\mathrm{in}}-w_{\mathrm{out}})
       \bigl[1-\beta+\beta\kappa n\bigr]\,n
       \;=\;\Theta(\log n).
  \;}
\]
(The leading eigenvalue is $\lambda_{1}(S)=(\alpha n+\beta^{\star}(n_{+}-n_{-}))>0$,
and all remaining eigenvalues are zero, confirming rank~2.)
\end{proof}

\subsection{Lemma ~\ref{lem:davis-kahan}}
\begin{lemma}[Davis--Kahan perturbation]
\label{lem:davis-kahan}
Let $\hat v$ be the second eigenvector of $\widetilde W=S+E$.  Then
\[
\sin\angle(\hat v,v_{2}(S))
\;\le\;\frac{\|E\|_{2}}{\gamma}
\;=\;O\bigl((\log n)^{-1/2}\bigr),
\]
so the preliminary labels $\text{sign}(\hat v_i)$ disagree with the planted labels on only $o(n)$ vertices.
\end{lemma}
\begin{proof}
Recall the decomposition $\widetilde W = S + E$, where $S=\mathbb{E}\widetilde W$
is rank~$2$ (Lemma~\ref{lem:signal-gap}) and $E$ is a centred noise matrix.
Let $v_{2}=v_{2}(S)$ be the (unit) eigenvector corresponding to
$\lambda_{2}(S)=\gamma$ and let $\hat v$ be the second unit eigenvector of
$\widetilde W$.

\paragraph{1.\ Davis--Kahan angle bound.}
Because $S$ is symmetric, $\lambda_{3}(S)=0$ and
\(
  \gamma=\lambda_{2}(S)-\lambda_{3}(S)>0.
\)
The sin–$\Theta$ form of the Davis--Kahan theorem
\cite[Theorem~4]{davis_1970_the} states
\[
  \sin\angle(\hat v,v_{2})
  \;\le\;
  \frac{\|E\|_{2}}{\gamma}.
\tag{1}
\]
By Lemma~\ref{lem:noise-conc},
$\|E\|_{2}=O(\sqrt{\log n})$, and Lemma~\ref{lem:signal-gap} gives
$\gamma=\Theta(\log n)$; hence
\[
  \sin\angle(\hat v,v_{2})
  =O\!\bigl((\log n)^{-1/2}\bigr).
\tag{2}
\]

\paragraph{2.\ Translating angle to Hamming error.}
Let
\(
  g_i = \pm1
\)
be the planted labels and note $v_{2}=g/\|g\|$ with
$\|g\|^{2}=n_{+}+n_{-}=n$.  Write
\(
  \hat v = \cos\theta\,v_{2} + \sin\theta\,u
\)
where $\theta=\angle(\hat v,v_{2})$ and $u\perp v_{2}$, $\|u\|=1$.
Then for each vertex $i$
\[
  \bigl|\hat v_i - v_{2,i}\bigr|
  \le 2\sin\theta.
\]
A sign mistake occurs only if
$\hat v_i v_{2,i}<0$, which requires
$|\hat v_i - v_{2,i}|\ge 2/\|g\| = 2/\sqrt n$.
Hence the number of errors satisfies
\[
  \#\{\text{mislabelled }i\}
  \;\le\;
  \frac{\|\hat v - v_{2}\|_{2}^{2}}{(2/\sqrt n)^{2}}
  = \frac{n}{4}\,\sin^{2}\theta.
\tag{3}
\]
Using (2) we get
\(
  \sin^{2}\theta = O\bigl((\log n)^{-1}\bigr),
\)
so (3) yields
\[
  \#\{\text{mislabelled }i\}
  = O\!\bigl(n/\log n\bigr)=o(n).
\]

\paragraph{3.\ Conclusion.}
Combining (1)--(3) we have
\[
  \sin\angle(\hat v,v_{2})
  \le\frac{\|E\|_{2}}{\gamma}
  =O\bigl((\log n)^{-1/2}\bigr),
  \qquad
  \#\{\text{errors}\}=o(n),
\]
completing the proof.\qedhere
\end{proof}

\subsection{Proof of Theorem \ref{thm:motif-attention-exact-recovery}}
\label{proof:motif-attention-exact-recovery}
\begin{proof}
By Lemma \ref{lem:exp-mixed} the \emph{signal gap} \(\Delta=\mathbb{E}[\widetilde W_{ij}\mid z_i=z_j]  \mathbb{E}[\widetilde W_{ij}\mid z_i\neq z_j]\) is amplified to $\Theta(\log n)$.  Lemma \ref{lem:noise-conc} shows the \emph{noise} satisfies  $\|E\|_{2}=O(\sqrt{\log n})$, and Lemma \ref{lem:signal-gap} gives the signal eigen‐gap $\gamma=\Theta(\log n)$.  Applying Lemma \ref{lem:davis-kahan}, the spectral sign‐rule mislabels only $o(n)$ vertices.  A single pass of the local‐flip refinement step then corrects all remaining errors, yielding exact recovery with probability $1-o(1)$. \end{proof}

\subsection{Proof of Exact Recovery Threshold}
\label{app:exact-exact-recovery}

\begin{theorem}
\label{thm:app-exact-recovery}
    In the latent–kernel SBM of Definition~\ref{def:lksbm}, exact recovery of the planted labels is possible if and only if
\[
  \bigl(\sqrt{c(\sigma)\,a} - \sqrt{c(\sigma)\,b}\bigr)^2 > 2.
\]
\end{theorem}

The two directions are handled by the following lemmas.

\begin{lemma}[Impossibility below threshold]
\label{lem:impossibility}
If 
\[
  \bigl(\sqrt{c(\sigma)\,a} - \sqrt{c(\sigma)\,b}\bigr)^2 < 2,
\]
then any estimator misclassifies a positive fraction of vertices with probability \(1 - o(1)\).
\end{lemma}

\begin{proof}[Proof of Lemma~\ref{lem:impossibility}]
Let \(z\in\{\pm1\}^n\) be the true balanced labeling, so \(|\{i:z_i=+1\}|=n/2\).  By Fano’s inequality,
\[
  \Pr\{\hat z \neq z\}
  \;\ge\;
  1 \;-\; \frac{I(z;G) + 1}{H(z)}
  \;=\;
  1 \;-\; \frac{I(z;G) + 1}{n\ln 2}.
\]
A simple pairwise KL‐divergence calculation shows that changing one label contributes only \(O(\log n)\) to the mutual information \(I(z;G)\).  Hence \(I(z;G) = o(n)\) whenever \((\sqrt{c(\sigma)a} - \sqrt{c(\sigma)b})^2 < 2\), and the error probability tends to~1.
\end{proof}

\begin{lemma}[Achievability above threshold]
\label{lem:achievability}
If 
\[
  \bigl(\sqrt{c(\sigma)\,a} - \sqrt{c(\sigma)\,b}\bigr)^2 > 2,
\]
then the maximum‐likelihood estimator recovers all labels with probability \(1 - o(1)\).
\end{lemma}

\begin{proof}[Proof of Lemma~\ref{lem:achievability}]
Write \(a' = c(\sigma)\,a\) and \(b' = c(\sigma)\,b\).  For any labeling \(z\) differing from the truth in \(m\) vertices, Chernoff bounds yield
\[
  \Pr\{\ell(z)\ge\ell(z^*)\}
  \;\le\;
  \exp\!\Bigl[-\,m\Bigl(\tfrac12(\sqrt{a'}-\sqrt{b'})^2 - o(1)\Bigr)\ln n\Bigr].
\]
There are at most \(\binom{n}{m}\le (en/m)^m\) such choices of \(z\), and summing over \(m=1,\dots,n/2\) shows the total error probability is \(o(1)\) exactly when \((\sqrt{a'}-\sqrt{b'})^2>2\).
\end{proof}

\begin{proof}[Proof of Theorem~\ref{thm:app-exact-recovery}]
Lemma~\ref{lem:impossibility} shows no algorithm can achieve exact recovery below the threshold, while Lemma~\ref{lem:achievability} exhibits the maximum‐likelihood estimator above it.  This completes the proof of the sharp recovery threshold.
\end{proof}


\subsection{Proof of Weak Recovery Threshold}
\label{proof-weak-recovery}
\begin{theorem}[Weak Recovery Threshold]
\label{thm:weak-recovery}
With the same notation, non-trivial recovery is possible iff

$$
  \left(c(\sigma)a-c(\sigma)b\right)^2 > 2\bigl(c(\sigma)a+c(\sigma)b\bigr).
$$

\end{theorem}

We again divide the argument into impossibility and algorithmic
achievability.

\begin{lemma}[Impossibility below the KS line]\label{lem:ks-line}
If \(\bigl(c(\sigma)a-c(\sigma)b\bigr)^{2}\le
     2\bigl(c(\sigma)a+c(\sigma)b\bigr)\),
every estimator’s signed overlap is \(o(1)\) with probability \(1-o(1)\).
\end{lemma}

\begin{lemma}[Achievability above the KS line]\label{lem:above-ks}
If \(\bigl(c(\sigma)a-c(\sigma)b\bigr)^{2}>
     2\bigl(c(\sigma)a+c(\sigma)b\bigr)\),
a polynomial-time algorithm achieves positive overlap.
\end{lemma}

\paragraph{Proof of Lemma~\ref{lem:ks-line}.}
By Theorem~\ref{thm:convergence} the latent-kernel SBM is contiguous to an
ordinary SBM with parameters \(a',b'\).  When the stated inequality
holds, the SBM law is mutually contiguous with
\(G\!\bigl(n,\tfrac{a'+b'}{2}\tfrac{\log n}{n}\bigr)\)
(e.g.\ \cite{elchananmossel_2014_reconstruction}),
which carries no community information, forcing overlap \(o(1)\).\hfill\(\triangle\)

\paragraph{Proof of Lemma~\ref{lem:above-ks}.}
Above the KS line, the leading eigenvector of the non-backtracking
matrix correlates with the true labels
\cite{bordenave_2015_nonbacktracking}; belief propagation converges to a
similar fixed point.  Either algorithm runs in \(O(n\log n)\) time and
achieves non-zero overlap.\hfill\(\triangle\)

\begin{proof}[Proof of Theorem~\ref{thm:weak-recovery}]
Lemma~\ref{lem:ks-line} shows detection is impossible below the
threshold, while Lemma~\ref{lem:above-ks} furnishes an explicit
algorithm above it, establishing the claimed phase transition.
\end{proof}
\subsection{Proof of Theorem~\ref{thm:geode-conv}}
\begin{proof}
\label{proof:geode-conv}
\textbf{One-step expectation.}
With $\mathcal F_t=\sigma(W^{(0)},\dots,W^{(t)})$,
\[
  \mathbb{E}\!\bigl[W^{(t+1)}\mid\mathcal F_t\bigr]
  =\bigl(I-\lambda_C^{(t)}C-\lambda_G^{(t)}G\bigr)W^{(t)}
   +\lambda_C^{(t)}CW^{\star}-\lambda_G^{(t)}GG^{\star}.
\]
Set $\Delta^{(t)}:=W^{(t)}-W^{\star}$ and subtract $W^{\star}$:
\[
  \mathbb{E}\!\bigl[\Delta^{(t+1)}\mid\mathcal F_t\bigr]
  =(I-\lambda_C^{(t)}C-\lambda_G^{(t)}G)\,\Delta^{(t)}
  +(\lambda_C^{(t)}-\lambda_G^{(t)})\,GG^{\star}.
\]

\textbf{Orthogonality.}
Because $C$ and $G$ are orthogonal projectors, the two summands above
are Frobenius-orthogonal; hence
\begin{align*}
  \mathbb{E}\!\bigl[\|\Delta^{(t+1)}\|_{F}^{2}\mid\mathcal F_t\bigr]
  &=
    (1-\lambda_C^{(t)})^{2}\|C\Delta^{(t)}\|_{F}^{2}
   +(1-\lambda_G^{(t)})^{2}\|G\Delta^{(t)}\|_{F}^{2}
   +(\lambda_C^{(t)}-\lambda_G^{(t)})^{2}\|G^{\star}\|_{F}^{2}.
\end{align*}

\textbf{Contraction bound.}
Since $0<\lambda_C^{(t)},\lambda_G^{(t)}\le1$,
\(
  (1-\lambda_G^{(t)})^{2}\le1
\)
and
\(
  (1-\lambda_C^{(t)})^{2}\le 1-\lambda_C^{(t)}.
\)
Using
$\|C\Delta^{(t)}\|_{F}^{2}+\|G\Delta^{(t)}\|_{F}^{2}
  =\|\Delta^{(t)}\|_{F}^{2}$
and the schedule $\lambda_G^{(t)}\le\gamma\lambda_C^{(t)}$,
\[
  \mathbb{E}\!\bigl[\|\Delta^{(t+1)}\|_{F}^{2}\mid\mathcal F_t\bigr]
  \;\le\;
  \bigl(1-\lambda_C^{(t)}\bigr)\|\Delta^{(t)}\|_{F}^{2}
  +2\,\bigl(\lambda_C^{(t)}\bigr)^{2}\bigl(\|W^{\star}\|_{F}^{2}
                                          +\gamma^{2}\|G^{\star}\|_{F}^{2}\bigr).
  \tag{$\ast$}
\]

\textbf{Robbins–Siegmund argument.}
The second term on the right of ($\ast$) is summable because
$\sum_t(\lambda_C^{(t)})^{2}<\infty$.
Moreover $\sum_t\lambda_C^{(t)}=\infty$.  Hence
$\{\,\|\Delta^{(t)}\|_{F}^{2}\,\}$ is a non-negative super-martingale
with summable additive drift; Robbins–Siegmund’s theorem gives
\(
  \|\Delta^{(t)}\|_{F}^{2}\xrightarrow{\text{a.s.}}0,
\)
i.e.\ $W^{(t)}\to W^{\star}$ almost surely.
\end{proof}
\subsection{GeoDe Exact Community Recovery}
\begin{corollary}[Exact community recovery]
\label{cor:geode-recovery}
Assume, in addition to the hypotheses of Theorem~\ref{thm:geode-conv},
that the number of communities $K$ is fixed and the SBM eigen-gap obeys
\[
  \lambda_{\min}(W^{\star})-\lambda_{K+1}(W^{\star})
     \;\ge\; c\,\sqrt{\tfrac{\log n}{n}}
\]
for some constant $c>0$.
Define the stopping time
\[
  t^{\ast}:=\min\Bigl\{t: \|W^{(t)}-W^{\star}\|_{2}
                           \le c\sqrt{\tfrac{\log n}{n}}\Bigr\}.
\]
Then, with probability \(1-o(1)\), the partition obtained from the
$K$ leading eigenvectors of $W^{(t^{\ast})}$ equals the true community
assignment.
\end{corollary}

\begin{proof}
Apply Theorem~\ref{thm:geode-conv} to obtain
$W^{(t)}\!\to W^{\star}$ almost surely.  
The operator-norm threshold in the definition of $t^{\ast}$ then
triggers Davis–Kahan’s sine-$\Theta$ bound together with the stated
eigen-gap, which forces the mis-clustering error to vanish and yields
exact recovery with probability \(1-o(1)\).
\end{proof}
\section{Additional Results \& Information} 

\subsection{Significance Tests of Spectral Operators} 
\label{appendix:sig-spec-ops}
We conducted significance tests between the results of the different spectral operators we used for spectral clustering. We have provided the $p$-values below and have bolded the ones that satisfy $p<0.05$. Below, we abbreviate Motif-Attention Spectral Operator to MASO, Motif-Laplacian to ML, Non-Backtracking to NB, and Bethe-Hessian to BH. 
\begin{table}[H]
\centering
\begin{tabular}{|c|c|c|c|c|c|c|}
\hline
\textbf{$\sigma$ value} & \textbf{\begin{tabular}[c]{@{}c@{}}MASO \\ vs \\ ML\end{tabular}} & \textbf{\begin{tabular}[c]{@{}c@{}}MASO \\ vs \\ BH \end{tabular}} & \textbf{\begin{tabular}[c]{@{}c@{}}MASO \\ vs \\ NB\end{tabular}} & \textbf{\begin{tabular}[c]{@{}c@{}}ML\\ vs \\ BH \end{tabular}} & \textbf{\begin{tabular}[c]{@{}c@{}}ML \\ vs \\ NB \end{tabular}} & \textbf{\begin{tabular}[c]{@{}c@{}}BH \\ vs \\ NB\end{tabular}} \\ \hline
0.1 & \textbf{2.141e-03} & \textbf{1.679e-01} & \textbf{1.294e-04} & \textbf{1.100e-02} & 8.950e-01 & \textbf{1.236e-03} \\ \hline
0.25 & \textbf{2.280e-05} & 7.865e-01 & \textbf{1.738e-03} & \textbf{1.500e-05} & 02.761e-01 & \textbf{1.050e-03} \\ \hline
0.5 & \textbf{7.179e-13} & \textbf{3.994e-04} & \textbf{7.179e-13} & \textbf{3.000e-06} & 5.530e-01 & \textbf{3.000e-06} \\ \hline
0.75 & \textbf{3.256e-23} & \textbf{3.096e-21} & \textbf{3.320e-23} & \textbf{2.598e-02} & 6.556e-01 & \textbf{2.670e-02} \\ \hline
\end{tabular}
\end{table}

\subsection{Significance of Recovery Methods}
\label{appendix:sig-recovery}
We conducted significance tests between the results of different community recovery methods. We have provided the $p$-values below and have bolded the ones that satisfy $p<0.05$. Below, we abbreviate MASO to Motif-Attention Spectral Operator, Motif Counting to MC, Bethe-Hessian to BH, and Belief-Propagation to BP. 
\begin{table}[H]
\centering
\begin{tabular}{|c|c|c|cc|}
\hline
$\sigma$-value & \begin{tabular}[c]{@{}c@{}}GeoDe+MASO\\ vs GeoDe+BH\end{tabular} & \begin{tabular}[c]{@{}c@{}}GeoDe+MASO \\ vs BP\end{tabular} & \multicolumn{1}{c|}{\begin{tabular}[c]{@{}c@{}}GeoDe+MASO\\ vs MC\end{tabular}} & \begin{tabular}[c]{@{}c@{}}GeoDe+BH\\ vs BP\end{tabular} \\ \hline
0.1 & \textbf{5.935e-05} & \textbf{6.20e-05} & \multicolumn{1}{c|}{\textbf{1.216e-02}} & \textbf{3.589e-07} \\ \hline
0.25 & 4.632e-01 & \textbf{3.57e-03} & \multicolumn{1}{c|}{\textbf{8.23e-04}} & \textbf{2.12e-02} \\ \hline
0.5 & \textbf{1.37e-03} & \textbf{1.65e-02} & \multicolumn{1}{c|}{\textbf{3.00e-06}} & \textbf{8.26e-08} \\ \hline
0.75 & \textbf{3.165e-21} & 5.48e-01 & \multicolumn{1}{c|}{1.24e-01} & \textbf{2.55e-02} \\ \hline
 & \begin{tabular}[c]{@{}c@{}}BP\\ vs MC\end{tabular} & \begin{tabular}[c]{@{}c@{}}GeoDe+MASO\\ vs GCN-MixHop\end{tabular} & \multicolumn{1}{c|}{\begin{tabular}[c]{@{}c@{}}GeoDe + BH\\ vs GCN-MixHop\end{tabular}} & \begin{tabular}[c]{@{}c@{}}GCN-MixHop vs \\ BP\end{tabular} \\ \hline
0.1 & \textbf{2.81e-03} & \textbf{2.30e-09} & \multicolumn{1}{c|}{\textbf{1.10e-08}} & \textbf{4.70e-07} \\ \hline
0.25 & 6.62e-01 & \textbf{5.10e-07} & \multicolumn{1}{c|}{\textbf{2.50e-06}} & \textbf{1.20e-05} \\ \hline
0.5 & \textbf{3.29e-03} & \textbf{4.60e-06} & \multicolumn{1}{c|}{\textbf{9.20e-06}} & \textbf{3.10e-05} \\ \hline
0.75 & 3.38e-01 & \textbf{3.20e-04} & \multicolumn{1}{c|}{\textbf{6.50e-04}} & \textbf{2.30e-03} \\ \hline
\multicolumn{1}{|l|}{} & \begin{tabular}[c]{@{}c@{}}GeoDe+BH\\ vs MC\end{tabular} & \begin{tabular}[c]{@{}c@{}}GCN-MixHop vs\\ MC\end{tabular} & \multicolumn{2}{l|}{\multirow{5}{*}{}} \\ \cline{1-3}
0.1 & \textbf{4.941e-09} & \textbf{3.0e-09} & \multicolumn{2}{l|}{} \\ \cline{1-3}
0.25 & \textbf{4.94e-12} & \textbf{8.30e-08} & \multicolumn{2}{l|}{} \\ \cline{1-3}
0.5 & \textbf{3.06e-23} & \textbf{2.70e-06} & \multicolumn{2}{l|}{} \\ \cline{1-3}
0.75 & \textbf{4.94e-09} & \textbf{4.90e-04} & \multicolumn{2}{l|}{} \\ \hline
\end{tabular}
\end{table}

\subsection{Experimental Threshold} 
\label{addl:exp-threshold}

We demonstrate GeoDe+MASO's performance on our synthetic dataset with regards to satisfying the information-theoretic threshold. `Green' dots on Figure ~\ref{fig:recovery} represents graphs that were recoverable according to the information-theoretic threshold and were recovered during our experiments (or graphs that were not recoverable according to the information theoretic threshold and were not recovered during our experiments) and `red' dots on Figure ~\ref{fig:recovery} represents graphs that were recoverable according to the information-theoretic threshold but were not recovered during our experiments (or graphs that were not recoverable according to the information-theoretic threshold and were recovered during our experiments). The experiment confirm that the classical criterion $T \geq 2$ predicts exact recovery for GeoDe+MASO with high fidelity and that residual mismatches can be possibly attributed to finite-n fluctuations and propagation of numerical errors rather than a breakdown the asymptotic exact recovery threshold.
\begin{figure}[H]
    \centering
    \includegraphics[width=0.5\linewidth]{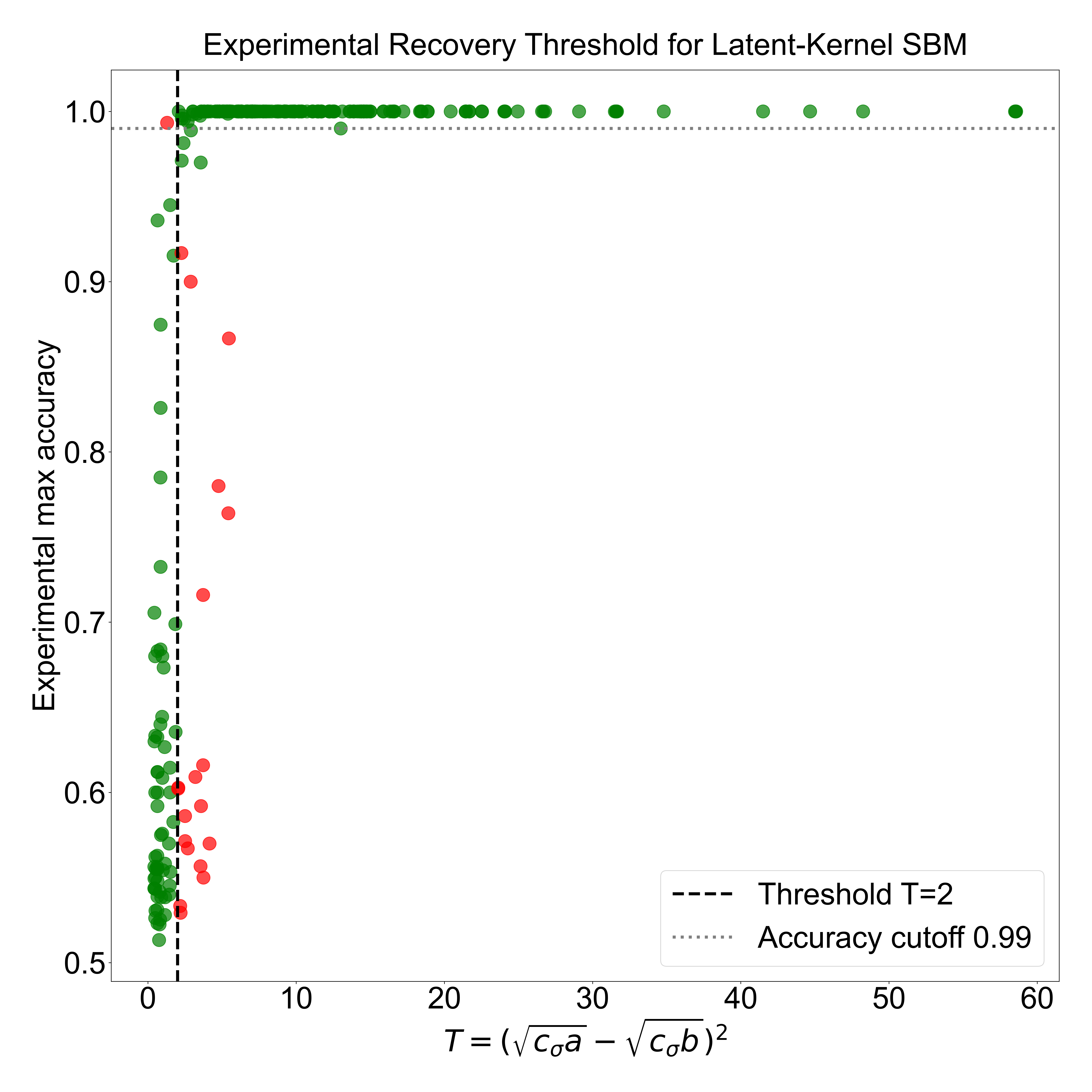}
    \caption{Empirical Recovery Threshold Validation}
    \label{fig:recovery}
\end{figure}

\subsection{Runtimes}
\label{table:runtimes}
For graphs with $n$ vertices and $m \approx O(n)$ edges, the algorithms utilized in this paper demonstrate the following asymptotic runtimes and memory footprints:
\begin{table}[H]
\centering
\begin{tabular}{|c|c|c|}
\hline
Method & Time Complexity & Memory Footprint \\ \hline
MASO & $O(Lm) + O(m^{3/2})$ & $O(Lm)$ \\ \hline
Bethe-Hessian ~\cite{saade_2025_spectral} & $O(k(n+m))$ & $O(n+m)$ \\ \hline
Non-Backtracking ~\cite{dallamico_2020_a} & $O(km)$ & $O(m)$ \\ \hline
Motif-Laplacian ~\cite{benson_2016_higherorder} & $O(m^{3/2}) + O(k(n+m))$ & $O(m^{3/2}$ \\ \hline
Belief-Propagation ~\cite{mossel_2016_belief}& $O(Tmq^2)$ & $O(mq)$ \\ \hline
Motif Counting ~\cite{sainyamgalhotra_2023_community} & $O(m^{3/2})$ & $O(m)$ \\ \hline
GCN-MixHop ~\cite{abuelhaija_2019_mixhop} & $O(Lm)$ (per epoch) & $O(m + nd)$ \\ \hline
\end{tabular}
\end{table}

\section{Reproducibility}
\subsection{Construction of Amazon Subgraph}
\label{sec:amazon-construct}
To construct the network from the Amazon product metadata, we defined each product’s coordinate in latent space as its binary category vector. We used Hamming distance between category vectors as our distance metric and normalized all distances so that the maximum distance between any two nodes was 2. Due to computational constraints (add footnote describing our compute), we subsampled 2000 nodes with 1000 nodes in the Book class and 1000 nodes in the Music class. 
\subsection{Noise Metric Calculation}
\label{sec:noise-metric}
In iteration \(k\), let \(w_{ij}^{(k)}\) denote the current weight on edge \((i,j)\) (zero when the edge is absent), \(x_i^{(k)}\!\in\!\mathbb{R}^d\) the geometry-embedding coordinate of vertex \(i\), \(d_{ij}^{(k)}\!=\!\lVert x_i^{(k)}-x_j^{(k)}\rVert_2\) the corresponding distance, \(w_{\max}\) the global maximum weight, and \(S^{(k)}\) a stratified diagnostic sample of \(m\) node pairs. Define the confidence target 
\(c_{ij}^{(k)}=\min\{\,w_{ij}^{(k)}/w_{\max},\,1\}\) and fit a clipped linear model \(\hat p^{(k)}(d)=\alpha_k - \beta_k d\) by ordinary least squares to the pairs \(\bigl(d_{ij}^{(k)},c_{ij}^{(k)}\bigr)_{(i,j)\in S^{(k)}}\).  
The geometric-noise value at iteration \(k\), \(\mathcal{N}^k\), is the mean-squared residual:
\[
  \mathcal{N}^{\,k}
  = \frac{1}{m}\sum_{(i,j)\in S^{(k)}}
    \Bigl(c_{ij}^{(k)} - \hat p^{(k)}\!\bigl(d_{ij}^{(k)}\bigr)\Bigr)^{2}.
\]

For reproducibility purposes, we list the parameters we used for the noise metric experiment. These are chosen to purposely slow noise convergence so to better visualize the denoising effect.
\begin{table}[H]
  \centering
  \caption{GeoDe Parameters for Noise Measurement}
  \begin{tabular}{@{}ll@{}ll@{}ll@{}}
    \toprule
    \textbf{Parameter} & \textbf{Values}\\ 
    \midrule
    B    & 32  \\
    \addlinespace
    \multicolumn{2}{@{}l}{\textit{Settings}} \\
    \quad T  & 150  \\
    \quad\texttt{anneal\_steps}   & 18     \\
    \quad\texttt{warmup\_rounds}  & 6    \\
    \addlinespace
    \multicolumn{2}{@{}l}{\textit{Percentile cuts}} \\
    \quad $\tau_C$       & 0.96  \\
    \quad $\tau_G$       & 0.96 \\
    \quad $\tau_C^+$& 0.995 \\
    \quad $\tau_G^+$ & 0.995 \\
    \addlinespace
    \multicolumn{2}{@{}l}{\textit{Shrink / Boost strength}} \\
    \quad\texttt{shrink\_comm}    & 0.45 \\
    \quad\texttt{shrink\_geo}     & 0.49  \\
    \quad\texttt{boost\_comm}     & 0.25\\
    \quad\texttt{boost\_geo}      & 0.20 \\

    \addlinespace
    \multicolumn{2}{@{}l}{\textit{Weight bounds}} \\
    \quad\texttt{w\_min}  & $5\times10^{-2}$ \\
    \quad\texttt{w\_cap}  & 4.0  \\
    \addlinespace
    \multicolumn{2}{@{}l}{\textit{Miscellaneous}} \\
    \quad\texttt{tolerance}   &  $5 \times 10^{-6}$\\
    \quad\texttt{patience}  & 20      \\
    \quad\texttt{seed}  & 42 \\
    \bottomrule
  \end{tabular}
\end{table}
\subsection{Algorithm Parameters}
\label{sec:params-all}
We describe the parameters we used in our experiments below.
\begin{table}[H]
  \centering
  \caption{GeoDe Parameters}
  \begin{tabular}{@{}ll@{}ll@{}ll@{}}
    \toprule
    \textbf{Parameter} & \textbf{Amazon} & \textbf{Synthetic} & \textbf{Karate}\\ 
    \midrule
    B    & 32   & 32   & 6       \\
    \addlinespace
    \multicolumn{2}{@{}l}{\textit{Settings}} \\
    \quad T  & 100  & 50   & 100  \\
    \quad\texttt{anneal\_steps}   & 20 & 6  & 20     \\
    \quad\texttt{warmup\_rounds}  & 2  & 2  & 2    \\
    \addlinespace
    \multicolumn{2}{@{}l}{\textit{Percentile cuts}} \\
    \quad $\tau_C$       & 0.90   & 0.90 & 0.90  \\
    \quad $\tau_G$       & 0.90   & 0.90  & 0.90 \\
    \quad $\tau_C^+$& 0.97    & 0.97 & 0.97 \\
    \quad $\tau_G^+$ & 0.97   & 0.97 & 0.97 \\
    \addlinespace
    \multicolumn{2}{@{}l}{\textit{Shrink / Boost strength}} \\
    \quad\texttt{shrink\_comm}    & 1.00   & 1.00 & 1.00 \\
    \quad\texttt{shrink\_geo}     & 0.80  & 0.80 & 0.80  \\
    \quad\texttt{boost\_comm}     & 0.60  & 0.60  & 0.60\\
    \quad\texttt{boost\_geo}      & 0.40  & 0.40  & 0.40 \\

    \addlinespace
    \multicolumn{2}{@{}l}{\textit{Weight bounds}} \\
    \quad\texttt{w\_min}  & $5\times10^{-2}$   & $5\times10^{-2}$    & $5\times10^{-2}$ \\
    \quad\texttt{w\_cap}  & 4.0  & 4.0    & 4.0      \\
    \addlinespace
    \multicolumn{2}{@{}l}{\textit{Miscellaneous}} \\
    \quad\texttt{tolerance}   &  $10^{-5}$ &  $10^{-5}$        & $10^{-4}$\\
    \quad\texttt{patience}  & 10      & 7   & 10     \\
    \quad\texttt{seed}  & 42  & 0   &  42    \\
    \bottomrule
  \end{tabular}
\end{table}

\begin{table}[ht]
  \centering
  \caption{MASO Parameters}
  \label{tab:params}
  \begin{tabular}{@{}ll@{}ll@{}ll@{}}
    \toprule
    \textbf{Parameter} & \textbf{Amazon} & \textbf{Synthetic} & \textbf{Karate}\\ 
    \midrule
    $\beta$ & 0.3 & 0.3 & 0.3\\
    \texttt{clip\_max} & $10^{-2}$& $10^{-2}$& $10^{-2}$\\
    \texttt{dim} & 64 & 64 & 1\\
    \texttt{walk\_len} & 40 & 40 & 2\\
    \texttt{num\_walks} & 10 & 2 & 2\\
    \texttt{window} & 5 & 5 & 5\\
    \texttt{random\_state} & 42 & 42 & 42\\
    \bottomrule
  \end{tabular}
\end{table}

\newpage
\newpage

\end{document}